\documentclass[11pt,letterpaper]{article}

\usepackage{geometry}
\usepackage[numbers]{natbib}
\usepackage{titlesec}
\usepackage{amsmath}
\usepackage{amsfonts}
\usepackage{amsthm}
\usepackage{amsbsy}
\usepackage{graphicx}
\usepackage{verbatim}
\usepackage{url}

\usepackage[T1]{fontenc}
\usepackage[tt=false, type1=true]{libertine}
\usepackage[varqu]{zi4}
\usepackage[libertine]{newtxmath}

\usepackage{bbm}
\usepackage{multirow}
\usepackage{enumitem}
\usepackage{xspace}
\usepackage{footnote}
\usepackage{hyperref}
\usepackage[svgnames]{xcolor}
\usepackage[capitalise,nameinlink]{cleveref}
\hypersetup{colorlinks={true},linkcolor={DarkBlue},citecolor=[named]{DarkGreen},urlcolor=black}
\usepackage{bm}
\usepackage{tikz}
\usepackage{algorithm}
\usepackage{algpseudocode}
\renewcommand{\algorithmiccomment}[1]{\bgroup\hfill\small\textcolor{gray}{//~#1}\egroup}
\algnotext{EndFor}
\algnotext{EndIf}
\algnotext{EndWhile}
\algnotext{EndProcedure}
\makeatletter
\renewcommand{\ALG@name}{Mechanism}
\makeatother

\usepackage{titling}
\usepackage[noblocks]{authblk}

\def\showauthnotes{1}
\ifthenelse{\showauthnotes=1}
{
	\newcommand{\authnote}[2]{{ \footnotesize \bf{[#1: #2]~}}}
}
{ \newcommand{\authnote}[2]{} }

\newtheorem{theorem}{Theorem}[section]
\newtheorem{lemma}[theorem]{Lemma}

\newtheorem{claim}[theorem]{Claim}
\newtheorem{fact}[theorem]{Fact}

\newtheorem{proposition}[theorem]{Proposition}
\newtheorem*{claim*}{Claim}

\theoremstyle{definition}

\newtheorem{definition}[theorem]{Definition}

\usepackage{soul}
\setuldepth{o}

\newcommand{\mms}{\ensuremath{\textrm{\MakeUppercase{mms}}}\xspace}
\newcommand{\pmms}{\ensuremath{\textrm{\MakeUppercase{pmms}}}\xspace}

\newcommand{\efo}{\ensuremath{\textrm{\MakeUppercase{ef}1}}\xspace}
\newcommand{\efx}{\ensuremath{\textrm{\MakeUppercase{efx}}}\xspace}
\newcommand{\ef}{\ensuremath{\textrm{\MakeUppercase{ef}}}\xspace}

{\list{}{\leftmargin=#1\rightmargin=#1}\item[]}%
{\endlist}

\DeclareMathOperator*{\argmin}{arg\,min}
\DeclareMathOperator*{\argmax}{arg\,max}

\allowdisplaybreaks

\usepackage[all]{nowidow}

\begin{document}
	
	\title{Round-Robin Beyond Additive Agents:\\ Existence and Fairness of Approximate Equilibria\thanks{\,This work was supported by the ERC Advanced 
			Grant 788893 AMDROMA ``Algorithmic and Mechanism Design Research in 
			Online Markets'', the MIUR PRIN project ALGADIMAR ``Algorithms, Games, 
			and Digital Markets'', and the NWO Veni project No.~VI.Veni.192.153.}
	}

	\author{Georgios Amanatidis}
	\affil[1]{Department of Mathematical Sciences\protect\\ University of Essex; Colchester, UK}
	\affil[ ]{{\textsf{\href{mailto:georgios.amanatidis@essex.ac.uk}{georgios.amanatidis}@essex.ac.uk}}\medskip}
	\author{Georgios Birmpas}
	\author[3]{Philip Lazos}
	\author{\\ Stefano Leonardi}
	\author[4]{Rebecca Reiffenh{\"a}user}
	\affil[2]{Department of Computer, Control, and Management Engineering \protect\\  Sapienza University of Rome; Rome, Italy}
	\affil[ ]{{\textsf{\{\href{mailto:birbas@diag.uniroma1.it}{birbas},
				\href{mailto:leonardi@diag.uniroma1.it}{leonardi}\}@diag.uniroma1.it}}\medskip}
	\affil[3]{Input Output; London, UK}
	\affil[ ]{{\textsf{\href{mailto:philip.lazos@iohk.io}{philip.lazos@iohk.io}}}\medskip}
	\affil[4]{Institute for Logic, Language and Computation\protect\\ University of Amsterdam; Amsterdam, The Netherlands}
	\affil[ ]{{\textsf{\href{mailto:r.e.m.reiffenhauser@uva.nl}{r.e.m.reiffenhauser}@uva.nl}}}
	\predate{}
	\postdate{}
	\date{}

	\maketitle

	\begin{abstract}
		\noindent Fair allocation of indivisible goods has attracted extensive attention over the last two decades, yielding numerous elegant algorithmic results and producing challenging open questions. The problem becomes much harder in the presence of \emph{strategic} agents. Ideally, one would want to design \emph{truthful} mechanisms that produce allocations with fairness guarantees. 
		However, in the standard setting without monetary transfers, it is generally impossible to have truthful mechanisms that provide non-trivial fairness guarantees.
		Recently, \citet{AmanatidisBFLLR21} suggested the study of mechanisms that produce fair allocations in their equilibria. Specifically, when the agents have additive valuation functions, the simple Round-Robin algorithm always has pure Nash equilibria and the corresponding allocations are \emph{envy-free up to one good} (\efo) with respect to the agents' \emph{true valuation functions}.
		Following this agenda, we show that this outstanding property of the Round-Robin mechanism extends much beyond the above default assumption of additivity. In particular, we prove that for agents with \emph{cancelable} valuation functions 
		(a natural class that 
		contains, e.g., additive and budget-additive functions), this simple mechanism always has equilibria and even its approximate equilibria correspond to approximately \efo allocations with respect to the agents' true valuation functions. 
		Further, we show that the approximate \efo fairness of approximate equilibria surprisingly holds for the important class of \emph{submodular} valuation functions as well, even though exact equilibria fail to exist!
	\end{abstract}

	\newpage
	
	\section{Introduction}
	\label{sec:intro}
	
	Fair division refers to the problem of dividing a set of resources among a group of agents in a way that every agent feels they have received a ``fair’’ share. The mathematical study of (a continuous version of) the problem dates back to the work of Banach, Knaster, and \citet{Steinhaus49}, who, in a first attempt to formalize fairness, introduced the notion of \emph{proportionality}, i.e., each of the $n$ agents receives at least $1/n$-th of the total value from fer perspective. Since then, different variants of the problem have been studied in mathematics, economics, political science, and computer science, and various fairness notions have been defined. The most prominent fairness notion is \emph{envy-freeness} \citep{GS58,Foley67,Varian74}, where each agent values her set of resources at least as much as the set of any other agent.
	When the available resources are \emph{indivisible} items, i.e., items that cannot be split among agents, notions introduced for infinitely divisible resources, like proportionality and envy-freeness are impossible to satisfy, even approximately. 
	In the last two decades fair allocation of indivisible items has attracted extensive attention, especially within the theoretical computer science community, yielding numerous elegant algorithmic results for various new fairness notions tailored to this discrete version of the problem, such as \emph{envy-freeness up to one good} (EF1) \citep{LMMS04,Budish11}, \emph{envy-freeness up to any good} (EFX) \citep{CaragiannisKMPS19}, and \emph{maximin share fairness} (MMS) \citep{Budish11}. We refer the interested reader to the surveys of \citet{Procaccia16-survey,BCM16-survey,AABFRLMVW22}.
	
	In this work, we study the problem of  fairly allocating indivisible \emph{goods}, i.e., items of non-negative value, to \emph{strategic} agents, i.e., agents who might misreport their private information if they have an incentive to do so. 
	Incentivising strategic agents to truthfully report their valuations is a central goal---and often a notorious challenge---in mechanism design, in general. 
	Specifically in fair division, this seems particularly necessary, since any fairness guarantee on the outcome of a mechanism typically holds with respect to its input, namely the \emph{reported} preferences of the agents rather than their true, private preferences which they may have chosen not to reveal. Without truthfulness, fairness guarantees seem to become  meaningless.
	Unfortunately, when monetary transfers are not allowed, as is the standard assumption in fair division, such \emph{truthful} mechanisms fail to exist for any meaningful notion of fairness, even for simple settings with two agents who have additive valuation functions \citep{ABCM17}.
	
	As an alternative, \citet{AmanatidisBFLLR21} initiated the study of \emph{equilibrium fairness}: when a mechanism always exhibits stable (i.e., pure Nash equilibrium) states, each of which corresponds to a fair allocation with respect to the \emph{true} valuation functions, the need for extracting agents' true preferences is mitigated. 
	Surprisingly, they show that for the standard case of additive valuation functions, the simple \emph{Round-Robin} routine is such a mechanism with respect to \efo fairness. Round-Robin takes as input an ordering of the goods for each agent, and then cycles through the agents and allocates the goods one by one, giving to each agent their most preferred available good. For agents with additive valuation functions, Round-Robin is known to produce \efo allocations (see, e.g., \citep{Markakis17-survey}). Note that, without monetary transfers, what distinguishes a mechanism from an algorithm is that its input is the, possibly misreported, agents' preferences.
	
	To further explore the interplay between incentives and fairness, we take a step back and focus solely on this very simple, yet fundamental, allocation protocol. It should be noted that the Round-Robin algorithm is one of the very few fundamental procedures one can encounter throughout the discrete fair division literature. Its central role is illustrated by various prominent results, besides producing \efo allocations: it can be modified 
	to produce approximate \mms allocations \citep{AMNS17}, as well as \efo allocations for \emph{mixed goods and chores} (i.e., items with negative value) \citep{AzizCIW22}. It produces \emph{envy-free} allocations with high 
	probability when the values are drawn from distributions \citep{ManurangsiS21}, it is used to produce a ``nice'' initial allocation as a subroutine in the state-of-the-art approximation algorithms for \emph{pairwise maximin share fair} (\pmms) allocations \citep{Kurokawa17} and  \efx allocations \citep{ANM2019}, it has the lowest communication complexity of any known fair division algorithm, and, most relevant to this work, it is the \emph{only} algorithm for producing fair allocations for more than two agents that, when viewed as a mechanism, is known to even have equilibria \citep{GW17}.

	We investigate the existence and the \efo guarantees of approximate pure Nash equilibria of the Round-Robin mechanism beyond additive valuation functions, i.e., when the goods already assigned to an agent potentially change how they value the remaining goods. 
	In particular, we are interested in whether anything can be said about classes that largely generalize additive functions, like \emph{cancelable} functions, i.e., functions where the marginal values with respect to any subset maintain the relative ordering of the goods, and  \emph{submodular} functions, i.e., functions capturing the notion of diminishing returns. 
	Although the stability and equilibrium fairness properties of Round-Robin have been visited before \citep{GW17,AmanatidisBFLLR21}, to the best of our knowledge, we are the first to study the problem for non-additive valuation functions and go beyond exact pure Nash equilibria. 
	Cancelable functions also generalize budget-additive, unit-demand, and multiplicative valuation functions \citep{BergerCFF22}, and recently have been of interest in the fair division literature as several results can be extended to this class \citep{BergerCFF22,akrami2022efx,caragiannis2022existence}. For similar reasons, cancelable functions seem to be a good pairing with Round-Robin as well, at least in the algorithmic setting (see, e.g., Proposition \ref{prop:RR_ef1_canc}). 
	
	Nevertheless, non-additive functions seem to be massively harder to analyze in our setting and come with various obstacles. 
	First, it is immediately clear that, even without strategic agents, the input of an ordinal mechanism implemented as a simultaneous-move one-shot game, like the Round-Robin mechanism we study here, can no longer capture the complexity of a submodular function
	(see also the relevant discussion in Our Contributions). 
	As a result, translating this sequential assignment to an estimate on the value of each agent's \emph{bundle} of goods, is not obvious.
	Lastly, and this applies to cancelable functions as well, assuming equilibria do exist and enough can be shown about the value of the assigned bundles to establish fairness, there is no reason to expect that any fairness guarantee will hold with respect to the true valuation functions, as the agents may misreport their preferences in an arbitrary fashion.

	\subsection{Contribution and Technical Considerations}
	We study the well-known Round-Robin mechanism (Mechanism \ref{alg:MRR}) for the
	problem of fairly allocating a set of indivisible goods to a set of strategic agents. 
	We explore the existence of approximate equilibria, along with the fairness guarantees that the corresponding allocations provide with respect to the  agents' true valuation functions. Qualitatively, we generalize the surprising connection between the stable states of this simple mechanism and its fairness properties to all approximate equilibria equilibria and for valuation functions as general as subadditive cancelable and submodular.
	In more detail, our main contributions can be summarized as follows:

	\begin{itemize}
		\item We show that the natural generalization of the \emph{bluff profile} of \citet{GW17} is an exact PNE that always corresponds to an \efo allocation, when agents have \emph{cancelable} valuation functions (Theorem \ref{thm:bluff_PNE_cancelable} along with Proposition \ref{prop:RR_ef1_canc}). Our proof is simple and intuitive and  generalizes  the results of \citet{GW17} and \citet{AmanatidisBFLLR21}. 
		\item For agents with submodular valuation functions, we show that there are instances where no $(3/4+\varepsilon)$-approximate PNE exists (Proposition \ref{thm:no_PNE}), thus creating a separation between the cancelable and the submodular cases. Nevertheless, we prove that an appropriate generalization of the bluff profile is a $1/2$-approximate PNE (Theorem \ref{thm:RReq}) that also produces an $1/2$-\efo allocation with respect to the true valuation functions (Theorem \ref{thm:RRguasub}). 
		\item We provide a unified proof that connects the factor of an approximate PNE with the fairness approximation factor of the respective allocation. In particular, any $\alpha$-approximate PNE results in a ${\alpha}/{2}$-\efo allocation for subadditive cancelable agents (Theorem \ref{thm:cancelable_ef1_n}), and in a ${\alpha}/{3}$-\efo allocation for submodular agents (Theorem \ref{thm:submod_ef1_n}). We complete the picture by providing lower bounds in both cases (Theorem \ref{thm:additive_fair_eq} and Proposition \ref{thm:oxs_lower_bound}), which demonstrate that our results are almost tight.
	\end{itemize}
	While this is not the first time Round-Robin is considered for non-additive agents, see, e.g., \citep{BouveretL14}, to the best of our knowledge, we are the first to study its fairness guarantees for cancelable and submodular valuation functions, independently of incentives.
	As a minor byproduct of our work, Theorem \ref{thm:RRguasub} and the definition of the bluff profile imply that, given \emph{value oracles} for the submodular functions, we can use Round-Robin as a subroutine to produce ${1}/{2}$-\efo allocations. 
	
	This also raises the question of whether one should allow a more expressive bid, e.g., a value oracle. While, of course, this is a viable direction, we avoid it here as it comes with a number of issues. Allowing the input to be exponential in the number of goods is already problematic, especially when simplicity and low communication complexity are two appealing traits of the original mechanism. Moreover, extracting orderings from value oracles would essentially result in a mechanism equivalent to ours (if the ordering of an agent depended only on \emph{her} function) or to a sequential game (if the orderings depended on all the functions) which is not what we want to explore here.
	Note that less information is not necessarily an advantage towards our goal. While this results in a richer space of equilibria, fairness guarantees are increasingly harder to achieve.

	As a final remark, all the algorithmic procedures we consider run in polynomial time, occasionally assuming access to value oracles, e.g., Algorithms \ref{alg:Grenam}, \ref{mech:Devnam}, \ref{alg:Greedy_for_i}. Although we do not consider computational complexity questions here, like how do agents compute best responses or how do they reach approximate equilibria, we do  consider such questions interesting directions for future work.

	\subsection{Further Related Work}
	
	The problem of fairly allocating indivisible goods to additive agents in the non-strategic setting has been extensively studied; for a recent survey, see \citet{AABFRLMVW22}. 
	Although the additivity of the valuation functions is considered a standard assumption, there are many works that explore richer classes of valuation functions.  Some prominent examples include the computation of \efo allocations for agents with general non-decreasing valuation functions \citep{LMMS04},  \efx allocations (or relaxations of \efx) under agents with cancelable valuation functions \citep{BergerCFF22,akrami2022efx,caragiannis2022existence} and subaditive valuation functions \citep{PR18, ChaudhuryGM21}, respectively, and approximate \mms allocations for submodular, XOS, and subadditive agents \citep{BarmanK20, GhodsiHSSY22}.
	
	Moving to the strategic setting, \citet{CKKK09} and \citet{MarkakisP11}
	were the first to consider the question of whether it is possible to have mechanisms that are truthful and fair at the same time, again assuming additive agents. \citet{ABCM17} resolved this question for two agents, showing there is no truthful mechanism with fairness guarantees under any meaningful fairness notion. As a result, subsequent papers considered truthful mechanism design under restricted valuation function classes \citep{HPPS,BabaioffEF21}.

	The stability of Round-Robin was first studied by \citet{GW17}, who proved that it always has PNE by using a special case of retracted result of \citet{BouveretL14} (this did not affect the former though; see \citep{AzizBLM17}). Finally, besides the work of \citet{AmanatidisBFLLR21} mentioned earlier, the fairness properties of Round-Robin under strategic agents have recently been studied by \citet{PV22}. Therein it is shown that Round-Robin, despite being non-truthful, satisfies a relaxation of truthfulness, as it is \emph{not obviously manipulable}.

	\section{Preliminaries}
	\label{sec:prelims}

	For $a\in \mathbb{N}$, let $[a]$ denote the set $\{1, 2, \ldots, a\}$. We will use $N = [n]$ to denote the set of agents and $M = \{g_1,\ldots, g_m\}$ to denote the set of goods. 
	Each agent $i\in N$ has a valuation function $v_i:2^M\to \mathbb{R}_{\ge0}$ over the subsets of goods. We assume that all $v_i$ are  \emph{normalized}, i.e., $v_i(\emptyset)=0$. We also adopt the shortcut $v_i(T\,|\,S)$ for 
	the \emph{marginal value} of a set $T$ with respect to a set $S$, i.e., $v_i(T\,|\,S)=v_i(T \cup S) - v(S)$.
	If $T = \{g\}$, we write $v
	_i(g\,|\,S)$ instead of $v(\{g\}\,|\,S)$.
	For each agent $i\in N$, we say that $v_i$ is 
	\begin{itemize}
		\item \emph{non-decreasing} (often referred to as \emph{monotone}), if $v_i(S) \le v_i(T)$ for any $S \subseteq T \subseteq M$.
		\item \emph{submodular}, if $v_i(g \,|\, S) \ge  v_i(g \,|\, T)$ for any $S \subseteq T \subseteq M$ and $g\notin T$.
		\item \emph{cancelable}, if $v_i(S\cup \{g\}) > v_i(T \cup \{g\}) \Rightarrow v_i(S) > v_i(T)$ for any $S, T \subseteq M$ and $g\in M\setminus (S\cup T)$.
		\item \emph{additive}, if $v_i(S\cup T) = v_i(S)+ v_i(T)$ for every $S, T \subseteq M$ with $S\cap T= \emptyset$. 
		\item \emph{subadditive}, if $v_i(S\cup T) \le v_i(S)+ v_i(T)$ for every $S, T \subseteq M$.
	\end{itemize}
	Throughout this work,  we only consider non-decreasing valuation functions, e.g., when we refer to submodular functions, we mean non-decreasing submodular functions. Note that although both submodular and (subadditive) cancelable functions are strict superclasses of additive functions, neither one is a superclass of the other.

	We will occasionally need an alternative characterization of submodular functions due to \citet{NemhauserWF78}.
	
	\begin{theorem}[\citet{NemhauserWF78}]\label{thm:SM}
		A function $v:2^M \rightarrow \mathbb{R}_{\ge0}$ is (non-decreasing) submodular if and only if we have
		$v(T) \le v(S) + \sum_{i \in T \setminus S} v(i\,|\,S)$,
		for all $S, T \subseteq M$. 
	\end{theorem}
	
	Also, the following lemma summarizes some easy observations about cancelable functions.
	
	\begin{lemma}\label{lem:cancelable}
		If $v:2^M \rightarrow \mathbb{R}_{\ge0}$ is cancelable, then 
		$v_i(S\cup R) > v_i(T \cup R) \Rightarrow v_i(S) > v_i(T) $, implying that 
		$v_i(S) \ge v_i(T) \Rightarrow v_i(S\cup R) \ge v_i(T \cup R)$, for any $S, T, R \subseteq M$, such that $R\subseteq M\setminus S\cup T$. In particular, $v_i(S) = v_i(T) \Rightarrow v_i(S\cup R) = v_i(T \cup R)$. 
	\end{lemma}
	\noindent Note that, for $S, T \subseteq M$, Lemma \ref{lem:cancelable} directly implies that $\argmax_{g\in T} v(g) \subseteq \argmax_{g\in T} v(g\,|\,S)$.

	Despite the fact that the agents have valuation functions, the mechanism we study (Mechanism \ref{alg:MRR}) is \emph{ordinal}, i.e., it only takes as input a \emph{preference ranking} from each agent. Formally, the preference ranking $\succ_i$, which agent $i$ reports, defines a total order on $M$, i.e.,  $g \succ_i g'$ implies that good $g$ precedes good $g'$ in agent $i$' declared preference ranking.\footnote{See the discussion after the statement of Mechanism \ref{alg:MRR} about why assuming that the reported preference rankings are total (rather than partial) orders is without loss of generality.}
	We call the vector of the agents' declared preference rankings, $\bm{\succ}\, = (\succ_1, \ldots, \succ_n)$, the \emph{reported profile} for the instance. So, while an instance to our problem is an ordered triple $(N, M, \mathbf{v})$, where  $\mathbf{v} = (v_1,\ldots, v_n)$ is a vector of the agents' valuation functions, the input to Mechanism \ref{alg:MRR} is $(N, M, \bm{\succ})$ instead.
	
	Note that $\succ_i$ may not reflect the actual underlying values, i.e., $g \succ_i g'$ does not necessarily mean that $v_i(g) > v_i(g')$ or, more generally, $v_i(g\,|\,S) > v_i(g'\,|\,S)$ for a given $S\subseteq M$. 
	This might be due to agent $i$ misreporting her preference ranking, or due to the fact that any single preference ranking is not expressive enough to fully capture all the partial orders induced by a submodular function. 
	Nevertheless, a valuation function $v_i$ does induce a \emph{true preference ranking} $\succcurlyeq^*_{i|S}$ for each set $S\subseteq M$, which is a partial order, i.e., $g\succcurlyeq^*_{i|S} g' \Leftrightarrow v_i(g\,|\,S) \ge v_i(g'\,|\,S)$ for all $g,g'\in M$. We use $\succ^*_{i|S}$ if the corresponding preference ranking is \emph{strict}, i.e., when $g\succcurlyeq^*_{i|S} g' \,\wedge\, g'\succcurlyeq^*_{i|S} g \,\Rightarrow\, g=g'$, for all  $g,g'\in M\setminus S$. For additive (and more generally, for cancelable) valuations, we drop $S$ for the notation and simply write $\succcurlyeq^*_i$ or $\succ^*_i$. Finally, for a total order $\succ$ on $M$ and a set $T\subseteq M$, we use $\mathrm{top}(\succ, T)$ to denote the ``largest'' element of $T$ with respect to $\succ$.

	\subsection{Fairness Notions}\label{subsec:fairness}
	A fair division mechanism produces an \emph{allocation} $(A_1,\ldots,A_n)$, where $A_i$ is the \emph{bundle} of agent $i$, which is a partition of $M$.  The latter corresponds to assuming no free disposal, namely all the goods must be allocated. 
	
	There are several different notions which attempt to capture which allocations are ``fair''. The most prominent such notion in the fair division literature has been \emph{envy-freeness} (\ef) \citep{GS58,Foley67,Varian74}, which has been the starting point for other relaxed notions, more appropriate for the indivisible goods setting we study here, as \emph{envy-freeness up to one good} (\efo) \citep{LMMS04,Budish11} and \emph{envy-freeness up to any good} (\efx) \citep{CaragiannisKMPS19}. Here we focus on \efo.
	
	\begin{definition}\label{def:EF-EFX}
		An allocation $(A_1,\ldots,A_n)$ is 
		\begin{itemize}[itemsep=2pt,topsep=6pt]
			\item $\alpha$-\textit{envy-free} ($\alpha$-\ef), if for every $i, j\in N$, $v_i(A_i) \geq \alpha \cdot v_i(A_j)$. \label{def:EF}
			\item $\alpha$-\textit{envy-free up to one good} ($\alpha$-\efo), if for every pair of agents $i, j\in N$, with $A_j\neq\emptyset$, there exists a good $g\in A_j$, such that
			$v_i(A_i) \geq \alpha \cdot v_i(A_j\setminus \{g\})$. \label{def:efo}
		\end{itemize}
	\end{definition}
	
	When for every agent $j\in N$ with $A_j\neq\emptyset$, we have $v_i(A_i) \geq \alpha \cdot v_i(A_j\setminus \{g\})$ for some good $g\in A_j$, we say that $(A_1,\ldots,A_n)$ is $\alpha$-\efo \emph{from agent $i$'s perspective}, even when the allocation is not $\alpha$-\efo!

	\subsection{Mechanisms and Equilibria}\label{subsec:mecanisms}

	We are interested in \emph{mechanisms} that produce allocations with  \efo guarantees. When \textit{no payments} are allowed, like in our setting, an allocation mechanism $\mathcal{M}$ is just an allocation algorithm that takes as input the agents' reported preferences. 
	In particular, Round-Robin, the mechanism of interest here, takes as input the reported   profile $\bm{\succ}$ and produces an allocation of all the goods. This distinction in terminology is necessary as the reported input may not be consistent with the actual valuation functions due to the agents' incentives. When the allocation returned by $\mathcal{M}(\bm{\succ})$ has some fairness guarantee, e.g., it is $0.5$-\efo, we will attribute the same guarantee to the reported profile itself, i.e., we will say that $\bm{\succ}$ is $0.5$-\efo. 
	
	We study  the fairness guarantees of the (approximate) pure Nash equilibria of Round-Robin. Given a preference profile $\bm{\succ} \,= ({\succ}_1, \ldots, {\succ}_n)$, we write
	$\bm{\succ}_{-i}$ to denote $({\succ}_1, \ldots, {\succ}_{i-1},  \allowbreak {\succ}_{i+1}, \ldots, {\succ}_n)$ and given a preference ranking ${\succ}'_{i}$ we use 
	$({\succ}'_i, \bm{\succ}_{-i})$ to denote the profile $({\succ}_1, \ldots, {\succ}_{i-1}, \allowbreak {\succ}'_{i}, \allowbreak {\succ}_{i+1}, \ldots, {\succ}_n)$. For the next definition we  abuse the notation slightly: given an allocation $(A_1,\ldots, \allowbreak A_n)$ produced by $\mathcal{M}(\bm{\succ})$, we write $v_i(\mathcal{M}(\bm{\succ}))$ to denote $v_i({A}_i)$; similarly for $\mathcal{M}({\succ}'_i, \bm{\succ}_{-i})$. 
	
	\begin{definition}
		Let $\mathcal{M}$ be an allocation mechanism and consider a preference profile
		$\bm{\succ} \,= ({\succ}_1, \ldots, \allowbreak {\succ}_n)$. We say that the total order ${\succ}_{i}$ is an \emph{$\alpha$-approximate best response} to $\bm{\succ}_{-i}$ if for every total order, i.e., permutation ${\succ}'_{i}$ of $M$, we have $ \alpha \cdot v_i(\mathcal{M}({\succ}'_i, \bm{\succ}_{-i}))\le v_i(\mathcal{M}(\bm{\succ}))$. 
		The profile $\bm{\succ}$ is an \emph{$\alpha$-approximate pure Nash equilibrium} (PNE) if, for each $i\in N$, ${\succ}_{i}$ is an $\alpha$-approximate best response to $\bm{\succ}_{-i}$.
	\end{definition}
	
	\noindent When $\alpha = 1$, we simply refer to best responses and exact PNE.

	\subsection{The Round-Robin Mechanism} \label{subsec:RR}
	We state Round-Robin as a mechanism (Mechanism \ref{alg:MRR}) that takes as input a reported profile $({\succ}_1, \ldots, {\succ}_n)$.
	For the sake of presentation, we assume that the agents in each \emph{round} (lines \ref{line:rr3}--\ref{line:rr6}) are always considered according to their ``name'', i.e., agent $1$ is considered first, agent $2$ second, and so on, instead of having a permutation determining the priority of the agents as an extra argument of the input. 
	This is without loss of generality, as it only requires renaming the agents accordingly.
	We often refer to the process of allocating a good to an agent (lines \ref{line:rr4}--\ref{line:rr6}) as a \emph{step} of the mechanism.
	
	\begin{algorithm}[ht]
		\caption{Round-Robin$({\succ}_1, \ldots, {\succ}_n)$ \hfill\small{ // For $i\in N$, ${\succ}_i$ is the reported preference ranking of  agent $i$.}}
		\begin{algorithmic}[1]
			\vspace{1pt}\State $S=M$\textbf{;} $(A_1,\dots,A_n) = (\emptyset,\ldots,\emptyset)$\textbf{;} $k = \lceil m/n\rceil$ 
			\For{$r = 1, \dots, k$} \Comment{Each value of $r$ determines the corresponding \ul{round}.}
			\For{$i = 1, \dots, n$} \Comment{The combination of $r$ and $i$ determines the corresponding \ul{step}.}\label{line:rr3}
			\State $g = \mathrm{top}(\succ_{i}, S)$ \label{line:rr4}
			\State $A_i = A_{i} \cup \{g\}$ \Comment{The current agent receives (what appears to be) her favorite available good.}
			\State $S = S\setminus \{g\}$ \Comment{The good is no longer available.}\label{line:rr6} \vspace{-2pt}
			\EndFor
			\EndFor
			\State 
			\Return {$(A_1,\dots,A_n)$}
		\end{algorithmic}
		\label{alg:MRR}
	\end{algorithm}

	Note that there is no need for a tie-breaking rule here, as the reported preference rankings are assumed to be total orders. Equivalently, one could allow for partial orders (either directly or via cardinal bids as it is done in \citep{AmanatidisBFLLR21}) paired with a deterministic tie-breaking rule, e.g., lexicographic tie-breaking, a priori known to the agents. 
	
	In the rest of the paper, we will assume that $m = k n$ for some $k\in \mathbb N$, for simplicity. Note that this is without loss of generality, as we may introduce at most $n-1$ dummy goods that have marginal value of $0$ with respect to any set for everyone and append them at the end of the reported preference rankings to be allocated during the last steps of the mechanism. 
	
	We have already mentioned that Round-Robin as an algorithm produces \efo allocations for additive agents, where the input is assumed to be any strict variant $\bm{\succ}^* \,= (\succ^*_{1|\emptyset}, \succ^*_{2|\emptyset}, \ldots, \succ^*_{n|\emptyset})$ of the truthful profile $(\succcurlyeq^*_{1|\emptyset}, \succcurlyeq^*_{2|\emptyset}, \ldots, \succcurlyeq^*_{n|\emptyset})$, i.e., the profile where each agent ranks the goods according to their singleton value. This property fully extends to cancelable valuation functions as well. The proof of Proposition \ref{prop:RR_ef1_canc} is rather simple, but not as straightforward as the additive case; note that it requires Lemma \ref{lem:cans_ord} from the next section. 
	
	\begin{proposition}\label{prop:RR_ef1_canc}
		Let be $\bm{\succ}^*$ be as described above. When all agents have cancelable valuation functions, the allocation returned by Round-Robin$(\bm{\succ}^*)$ is \efo.
	\end{proposition}
	
	\begin{proof}
		Let $(A_1,\dots,A_n)$ be the allocation returned by Round-Robin$(\bm{\succ}^*)$. Fix two agents, $i$ and $j$, and let $A_i = \{x_1, x_2, \ldots, x_k\}$ and $A_j = \{y_1, y_2, \ldots, y_k\}$, where the goods in both sets are indexed according to the round in which they were allocated to $i$ and $j$, respectively. 
		By the way Mechanism \ref{alg:MRR} is defined, we have $x_r \succ^*_{i|\emptyset} y_{r+1}$, for all $r \in [k-1]$. Therefore, $x_r \succcurlyeq^*_{i|\emptyset} y_{r+1}$, or equivalently, $v_i(x_r) \ge v_i(y_{r+1})$, for all $r \in [k-1]$.
		Thus, by Lemma \ref{lem:cans_ord}, we get $v_i(A_i\setminus \{x_k\}) \ge v_i(A_j\setminus \{y_1\})$, and using the fact that $v_i$ is non-decreasing, $v_i(A_i) \ge v_i(A_j\setminus \{y_1\})$.
	\end{proof}

	\section{Existence of approximate PNE}
	\label{sec:a}
	At first glance, it is not clear why Mechanism \ref{alg:MRR} has any pure Nash equilibria, even approximate ones for a constant approximation factor. For additive valuation functions, however, it is known that for any instance we can construct a simple preference profile, called the \emph{bluff profile}, which is an exact PNE. While the proof of this fact, in its full generality, is fragmented over three papers \citep{GW17,BL16,AmanatidisBFLLR21}, we give here a simple proof that generalizes the existence of exact PNE to cancelable valuation functions. As we shall see later, extending this result to submodular functions is not possible and even defining a generalization of the bluff profile which is a $0.5$-approximate PNE is not straightforward. 
	
	\subsection{Cancelable valuations}
	\label{subsec:cancelable}
	Defining the bluff profile for cancelable agents, we will start from a strict variant of the truthful profile $(\succcurlyeq^*_{1|\emptyset}, \succcurlyeq^*_{2|\emptyset}, \ldots, \succcurlyeq^*_{n|\emptyset})$, i.e., the profile where each agent ranks the goods according to their value (as singletons) in descending order, as we did for Proposition \ref{prop:RR_ef1_canc}. 
	Assume that any ties are broken deterministically to get the strict version $\bm{\succ}^* \,= (\succ^*_{1|\emptyset}, \succ^*_{2|\emptyset}, \ldots, \succ^*_{n|\emptyset})$. Now, consider $\textrm{Round-Robin}(\bm{\succ}^*)$ and let $h_1, h_2, \ldots, h_m$ be a renaming of the goods according to the order in which they were allocated and $\succ^\mathrm{b}$ be the corresponding total order (i.e., $h_1 \succ^\mathrm{b} h_2 \succ^\mathrm{b} \ldots \succ^\mathrm{b} h_m$). The \emph{bluff profile} is the preference profile $\bm{\succ}^\mathrm{b} \,= (\succ^\mathrm{b}, \succ^\mathrm{b}, \ldots, \succ^\mathrm{b})$, where everyone ranks the goods in the order they were allocated in Round-Robin$(\bm{\succ}^*)$. The following fact follows directly from the definition of the bluff profile and the description of Round-Robin. 
	
	\begin{fact}\label{fact:bluff}
		If $(\bm{\succ}^*)$ is a strict version of the truthful preference profile and $(\bm{\succ}^\mathrm{b})$ is the corresponding bluff profile, 
		then $\mathrm{Round\text{-}Robin}(\bm{\succ}^\mathrm{b})$ and $\mathrm{Round\text{-}Robin}(\bm{\succ}^*)$ both return the same allocation.
	\end{fact}
	
	An interesting observation about this fact is that, combined with Proposition \ref{prop:RR_ef1_canc} and Theorem \ref{thm:bluff_PNE_cancelable}, it implies that there is at least one PNE of Mechanism \ref{alg:MRR} which is \efo! 
	Of course, it is now known that all exact PNE of Round-Robin are \efo for agents with \emph{additive} valuation functions and, 
	as we will see later on, even approximate PNE have (approximate) \efo guarantees for much more general instances, including the case of \emph{subadditive cancelable} valuation functions.
	
	\begin{theorem}\label{thm:bluff_PNE_cancelable}
		When all agents have cancelable valuation functions, the bluff profile is an exact PNE of Mechanism \ref{alg:MRR}.
	\end{theorem}
	
	We first need to prove the following lemma that generalizes a straightforward property of additive functions for cancelable functions.
	
	\begin{lemma}\label{lem:cans_ord}
		Suppose that $v(\cdot)$ is a cancelable valuation function. Consider sets $X=\{x_1, x_2, \ldots, x_k\}$ and $Y=\{y_1, y_2, \ldots, y_k\}$. 
		If for every $j \in [k]$, we have that  $v(x_j) \geq v(y_j)$, then $v(X) \geq v(Y)$.
	\end{lemma}
	
	\begin{proof}
		We begin by arguing that it is without loss of generality to first assume that the elements of $X$ are ordered by non-increasing value with respect to $v$ and then also assume that $y_j\notin \{x_1, x_2, \ldots, x_{j-1}\}$, for any $j \in [k]$. The former is indeed a matter of reindexing, if necessary, the elements of $X$ and consistently reindexing the corresponding elements of $Y$. For the latter, 
		suppose that there exist $j$ such that $y_j = x_t$ for $t \le j-1$ and consider the smallest $t$ for which this happens. We have $v(x_t) \ge v(x_{t+1}) \ge \ldots \ge v(x_j)$ by the assumption on the ordering of the elements of $X$,  $v(x_j) \ge v(y_j)$ by hypothesis, and $v(y_j) = v(x_t)$. Thus, $v(x_t) = v(x_{t+1}) = \ldots =  v(x_j)$. Now we may rename the elements of $Y$ to $\{y'_1,  \ldots, y'_k\}$ by inserting $y_j$ to the $t$-th position, i.e., $y'_t = y_j$, $y'_{s} = y_{s-1}$, for $t+1 \le s \le j$, and $y'_{s} = y_{s}$, for $s < t$ or $s > j$. Since only $y_t, y_{t+1}, \ldots, y_j$ changed indices but $v(x_t) = v(x_{t+1}) = \ldots =  v(x_j)$, we again have that $v(x_j) \geq v(y'_j)$ for every $j \in [k]$. Moreover, now the  smallest $\ell$ for which there exist $j>\ell$ such that $y_j = x_\ell$ is strictly larger than $t$. By repeating this renaming of the elements of $Y$ we end up with a renaming $\{y^*_1,  \ldots, y^*_k\}$ such that for every $j \in [k]$, $v(x_j) \geq v(y^*_j)$ and $y^*_j\notin \{x_1, x_2, \ldots, x_{j-1}\}$.

		So, assuming that the elements of $X$ are ordered in non-increasing value with respect to $v$ and that $y_j\notin \{x_1, x_2, \ldots, x_{j-1}\}$, for any $j \in [k]$, suppose towards a contradiction that $v(X) < v(Y)$. That is,  $v(\{x_1, x_2, \ldots, x_k\}) \allowbreak < v(\{y_1, y_2, \allowbreak \ldots, y_k\})$. Observe that if $v(\{x_1, x_2, \ldots,x_{k-1}\}) \geq v(\{y_1, y_2, \ldots,y_{k-1}\})$, this would imply that $v(\{x_1,  \ldots,x_{k-1}, y_k\}) \geq v(\{y_1,  \ldots,y_{k-1}, y_k\})$, by the definition of cancelable valuations and the fact that $y_k \notin \{x_1, \ldots, x_{k-1}\} \cup \{y_1, \ldots, y_{k-1}\}$.  This leads to 
		\[
		v(\{x_1,  \ldots,x_{k-1}, x_k\}) \geq v(\{x_1,  \ldots,x_{k-1}, y_k\}) \geq v(\{y_1,  \ldots,y_{k-1}, \allowbreak y_k\})\,,
		\]
		where the first inequality follows from  $v(x_k) \geq v(y_k)$ and Fact \ref{lem:cancelable}, contradicting our initial assumption. Therefore, $v(\{x_1,  \ldots,x_{k-1}\})< v(\{y_1, \ldots,y_{k-1}\})$. By repeating the same argument $k-2$ more times, we end up with $v(x_1)<v(y_1)$, a contradiction.
	\end{proof}
	
	\begin{proof}[Proof of Theorem \ref{thm:bluff_PNE_cancelable}]
		Now we show that the bluff profile for cancelable valuations is an exact PNE. 
		Consider the goods named $h_1,\dots, h_m$ as in the bluff profile, i.e., by the order in which they are picked when each agent reports their preference order to be the one induced by all singleton good values.
		Consider agent $i$.
		Her assigned set of goods under the bluff profile is $A_i^\mathrm{b}= \{ h_i, h_{n+i},\dots, h_{(k-1)n+i}\}$, where $k = m/n$. 
		Assume now that she deviates from $\succ^\mathrm{b}$ to $\succ_i$, resulting in some allocated set $A_i=\{y_1, y_2,\dots, y_k\}$, where we assume $y_r$ to be allocated in round $r$. We need to show $v_i(A_i^\mathrm{b})\geq v_i(A_i)$.
		
		To this end, we compare the goods allocated to agent $i$ in both reports, one by one. If $v_i(y_r) \le v_i(h_{(r-1)n+i})$ for every $r \in [k]$, then we are done by applying Lemma \ref{lem:cans_ord} with $A_i^\mathrm{b}$ and $A_i$.
		If some of these inequalities fail, let  $r$ denote the latest round such that $v_i(y_r)>v_i(h_{(r-1)n+i}$.
		Therefore, in the execution of Mechanism \ref{alg:MRR} with the  bluff profile as input, $y_r$ was no longer available in round $r$. However, $y_r$ becomes available in round $r$ once agent $i$ deviates. This can only stem from the fact that at some point before round $r$, a good $h_t$ with $t>(r-1)n+i$ was picked (since the overall number of goods picked per round always stays the same). 
		Clearly, the only agent who could have done so (since she is the only one deviating from the common bluff order) is agent $i$. Therefore, it holds that $h_t=y_j$ for some $j<r$.
		Now, we replace the ordered set $Y = (y_1, y_2, \dots, y_k)$ by $Y' = (y_1, \dots, y_{j-1}, y_r, y_{j+1},\dots, y_{r-1}, y_j, y_{r+1},\dots, y_k)$, i.e., we simply exchange $y_r$ and $y_j$. 
		It will be convenient to rename $y_1, \ldots, y_k$ so that $Y' = (y'_1, y'_2, \dots, y'_k)$
		
		We claim that it if agent $i$ reports a preference ranking $\succ'_i$ that starts with all goods in $Y'$, in that specific order, followed by everything else, in any order, she still gets $A_i$ but the goods are allocated in the order suggested by $Y'$. 
		Indeed, first notice that the first $j-1$ rounds of Round-Robin will be the same as in the run with the original deviation $\succ_i$. Further, $y'_j=y_r$ is allocated earlier under $\succ'_i$ than under $\succ_i$, and thus it surely is available at the time. 
		After that, rounds $j-1$ to $r-1$ will be the same as in the run with the deviation $\succ_i$.
		Now $y'_r = y_j$ is allocated later than before, namely in round $r$, but it is not among the first $(r-1)n+i$ goods in the bluff order, as noted above, which means it is not allocated to any other agent in any round before the $r$-th under $\succ'_i$. Finally, rounds $r+1$ to $k$ will be the same as in the run with $\succ_i$.
		
		Although agent $i$ still is assigned the same set $A_i$ by deviating to $\succ'_i$, we now have $v_i(y'_r)=v_i(y_j) \le v_i(h_{(r-1)n+i}$, where the inequality holds because both goods are available in round $r$ of the bluff run, and agent one prefers $h_{(r-1)n+i}$. Also, all later goods in $Y'$ remain unchanged, i.e., $y'_{s} =y_{s}$ for $s > r$. Therefore, the latest occurrence of some $y'_{\ell} > h_{(\ell-1)n+i}$ now happens at an earlier point in the sequence, if at all.
		Repeating this process until no such occurrence is left yields an ordering $Y^* = (y^*_1, y^*_2, \dots, y^*_k)$ of $A_i$ such that for all $r\in [k]$, $v_i(y^*_r)\leq v_i(h_{(r-1)n+i})$. Now using Lemma \ref{lem:cans_ord} completes the proof.
	\end{proof}

	\subsection{Submodular valuations}
	\label{subsec:submod}
	We move on to the much more general class of submodular valuations. 
	In order to define the bluff profile in this case, we again would like to start from the truthful profile. However, recall that Round-Robin restricts each agent's report to specifying an ordering on the good set $M$ and these preference rankings are not expressive enough to fully capture submodular valuation functions. In fact, it is not obvious what `truthful' means here without further assumptions on what information is known by the agents. 
	Still, we define a \textit{truthfully greedy} allocation and use this as our starting point.
	
	Imagine that, instead of having a full preference profile from the beginning, we only ask the active agent $i$ (i.e., the agent to which we are about to allocate a new good) for the good with the largest marginal value with respect to her current set of goods $A_i$ and give this to her. Let $h_1, h_2, \ldots, h_m$ be a renaming of the goods according to the order in which they would be allocated in this hypothetical truthfully greedy scenario and $\succ^\mathrm{b}$ be the corresponding total order. Like in the cancelable case, the bluff profile is the preference profile $\bm{\succ}^\mathrm{b} \,= (\succ^\mathrm{b}, \succ^\mathrm{b}, \ldots, \succ^\mathrm{b})$.
	
	Formally, the renaming of the goods is performed as described in Algorithm \ref{alg:Grenam} below. It should be noted that this definition of the bluff profile is consistent with the definition for cancelable functions, assuming that all ties are resolved lexicographically. 
	
	\makeatletter
	\renewcommand{\ALG@name}{Algorithm}
	\makeatother
	
	\begin{algorithm}[ht]
		\caption{Greedy renaming of goods for defining the bluff profile 
			\\ Input: $N$, $M$, value oracles for $v_1(\cdot), \ldots, v_n(\cdot)$
		}\label{alg:Grenam}
		\begin{algorithmic}[1]
			\vspace{1pt}\State $X_i=\emptyset$ for $i \in [n]$
			\For{$j = 1, \dots, m$}
			\State $i = (j-1) \!\pmod n + 1$
			\State $h_j = \displaystyle \argmax_{g \in M\setminus \bigcup_{\ell} X_{\ell}}v_i(g \,|\, X_i)$ \Comment{Ties are broken lexicographically.} \smallskip
			\State $X_i=X_i\cup \{h_j\}$
			\EndFor \vspace{-1pt}
			\State \Return $(h_1, h_2, \ldots, h_m)$
		\end{algorithmic}
	\end{algorithm}

	Also notice that the allocation $\mathrm{Round\text{-}Robin}(\bm{\succ}^\mathrm{b})$ produced under the bluff profile is exactly $(X_1, X_2, \allowbreak\ldots,\allowbreak X_n)$, as described in Algorithm \ref{alg:Grenam}, i.e., $X_i = A_i^\mathrm{b}= \{ h_i, h_{n+i},\dots, h_{(k-1)n+i}\}$, where recall that $k = m/n$.

	The main result of this section is Theorem \ref{thm:RReq} stating that the bluff profile is a $\frac{1}{2}$-approximate PNE when agents have submodular valuation functions. While this sounds weaker than Theorem \ref{thm:bluff_PNE_cancelable}, it should be noted that for submodular agents Mechanism \ref{alg:MRR} does not have PNE in general, even for relatively simple instances, as stated in Proposition \ref{thm:no_PNE}. In fact, even the existence of approximate equilibria can be seen as rather surprising, given the generality of the underlying valuation functions.  
	
	\begin{proposition}\label{thm:no_PNE}
		There exists an instance where all agents have submodular valuation functions such that Mechanism \ref{alg:MRR} has no $(\frac{3}{4} +\varepsilon)$-approximate PNE.
	\end{proposition}

	\begin{proof}
		Consider an instance with 2 agents and 4 goods $M = \{g_1, g_2, g_3, g_4\}$, with the following valuation for all possible 2-sets:
		\begin{center}
			\vspace{-2ex}\begin{minipage}{0.31\textwidth}
				$$v_1(\{g_1, g_2\}) = 3$$
				$$v_1(\{g_1, g_3\}) = 3$$
				$$v_1(\{g_1, g_4\}) = 4$$
				$$v_1(\{g_2, g_3\}) = 4$$
				$$v_1(\{g_2, g_4\}) = 3$$
				$$v_1(\{g_3, g_4\}) = 3$$
			\end{minipage}
			\begin{minipage}{0.31\textwidth}
				$$v_2(\{g_1, g_2\}) = 4$$
				$$v_2(\{g_1, g_3\}) = 4$$
				$$v_2(\{g_1, g_4\}) = 3$$
				$$v_2(\{g_2, g_3\}) = 3$$
				$$v_2(\{g_2, g_4\}) = 4$$
				$$v_2(\{g_3, g_4\}) = 4$$
			\end{minipage}
		\end{center}\smallskip
		In addition, all individual goods have the same value: $v_1(x) = v_2(x) = 2$ for $x \in M$, while all $3$-sets and $4$-sets have value $4$, for both agents.
		
		We begin by establishing that this valuation function is indeed submodular for both agents. Observe for any set $S \subseteq M$ and $i\in[2], j\in[4]$ we have:\
		\begin{align*}
			|S| = 0 &\Rightarrow v_i(g_j \;|\; S) \in \{2\}\\[0.5ex]
			|S| = 1 &\Rightarrow v_i(g_j \;|\; S) \in \{1, 2\}\\[0.5ex]
			|S| = 2 &\Rightarrow v_i(g_j \;|\; S) \in \{0, 1\}\\[0.5ex]
			|S| = 3 &\Rightarrow v_i(g_j \;|\; S) = 0 \,,
		\end{align*}
		which immediately implies that both valuation functions are indeed submodular.
		
		Notice that for any reported preferences ${\succ}_1, {\succ}_2$, one of the two agents will receive goods leading to a value of $3$. If this is the agent $1$, she can easily deviate and get $4$ instead. In particular, if agent $2$ has good $g_2$ or $g_3$ first in their preferences then agent $1$ can get $\{g_1, g_4\}$, and if agent $2$ has good $g_1$ or $g_4$ as first then agent $1$ can get $\{g_2, g_3\}$ instead. On the other hand, if agent $2$ received a value of $3$ they can also always deviate to $4$. Notice that for any $g_a$, agent $2$ always has two sets different sets $\{g_a, g_b\},\{g_a, g_c\}$ with value $4$ and one $\{g_a, g_d\}$ with value 3. Thus, for any preference of agent $1$ with $g_{\hat a} \succ_1 g_{\hat b} \succ_1 g_{\hat c} \succ_1 g_{\hat d}$, agent 2 can deviate and get either $\{g_{\hat b}, g_{\hat d}\}$ or $\{g_{\hat c}, g_{\hat d}\}$, one of which must have value $4$. Therefore, in every outcome there exists an agent that can deviate to improve their value from $3$ to $4$.
	\end{proof}

	Moving towards the proof of Theorem \ref{thm:RReq} for the submodular case, we note that although it is very different from that of Theorem \ref{thm:bluff_PNE_cancelable}, we will still need an analog of the main property therein, i.e., the existence of a good-wise comparison between the goods an agent gets under the bluff profile and the ones she gets by deviating. As expected, the corresponding property here (see Lemma \ref{lem:DecMansub}) is more nuanced 
	and does not immediately imply Theorem \ref{thm:RReq} as we are now missing the analog of Lemma \ref{lem:cans_ord}. 
	
	Throughout this section, we are going to argue about an arbitrary agent $i$. To simplify the notation, let us rename $X_i = A_i^\mathrm{b}= \{ h_i, h_{n+i},\dots, h_{(k-1)n+i}\}$ to simply $X=\{x_1, x_2,\ldots,x_k\}$, where we have kept the order of indices the same, i.e., $x_j = h_{(j-1)n+i}$. This way, the goods in $X$ are ordered according to how they were allocated to agent $i$ in the run of Mechanism \ref{alg:MRR} with the bluff profile as input. 
	
	We also need to define the ordering of the goods agent $i$ gets when she deviates from the bluff bid $\succ^\mathrm{b}$ to another preference ranking $\succ_i$. Let $A_i = Y= \{y_1, y_2, \ldots,y_k\}$ be this set of goods. Instead of renaming the elements of $Y$ in a generic fashion like in the proof of Theorem \ref{thm:bluff_PNE_cancelable}, doing so becomes significantly more complicated, and we need to do it in a more systematic way, see  Algorithm \ref{mech:Devnam}.
	
	\begin{algorithm}[ht]
		\caption{Greedy renaming of goods for the deviating agent $i$ 
			\\ Input: $X=\{x_1, x_2,\ldots,x_k\}$, $Y$, and a value oracle for $v_i(\cdot)$}\label{mech:Devnam}
		\begin{algorithmic}[1]
			\vspace{2pt}\State $Z=Y$
			\For{$j = |Y|, \dots, 1$}
			\State $y'_j = \displaystyle \argmin_{g \in Z}v_i(g \,|\, \{x_1,\ldots, x_{j-1}\})$ \Comment{Ties are broken lexicographically.}  \smallskip
			\State $Z=Z\setminus \{y'_j\}$ \vspace{-1pt}
			\EndFor
			\State \Return $(y'_1, y'_2, \ldots, y'_{|Y|})$
		\end{algorithmic}
	\end{algorithm}
	
	In what follows, we assume that the indexing $y_1, y_2, \ldots,y_k$ is already the result of Algorithm \ref{mech:Devnam}. This renaming is crucial and it will be used repeatedly. In particular, we need this particular ordering in order to prove that $v_i(x_j \,|\, \{x_1,\ldots, x_{j-1}\}) \geq v_i(y_j \,|\, \{x_1,\ldots, x_{j-1}\})$, for all $j\in [k]$, in Lemma \ref{lem:DecMansub} below. 
	Towards that, we need to fix some notation for the sake of readability. For $j\in [k]$, we use $X^j_-$ and $X^j_+$ to denote the sets $\{x_1, x_2, \ldots, x_{j}\}$ and $\{x_j, x_{j+1}, \ldots, x_{k}\}$, respectively.  The sets $Y^j_-$ and $Y^j_+$, for $j\in [k]$, are defined analogously. We also use $X^0_- = Y^0_- = \emptyset$.
	The main high-level idea of the proof is that if $v_i(y_{\ell} \,|\, X^{\ell-1}_-) > v_i(x_{\ell} \,|\, X^{\ell-1}_-)$ for some $\ell$, then it must be the case that during the execution of Round-Robin$(\bm{\succ}^\mathrm{b})$ every good in $Y^{\ell}_- = \{y_1, \ldots, y_{\ell}\}$ is allocated before the turn of agent $i$ in  round $\ell$. Then, using a simple counting argument, we show that agent $i$ cannot receive all the goods in $Y^{\ell}_-$ when deviating, leading to a contradiction.

	\begin{lemma}\label{lem:DecMansub}
		Let $X=\{x_1, x_2,\ldots,x_k\}$ be agent $i$'s bundle in Round-Robin$(\bm{\succ}^\mathrm{b})$, where goods are indexed in the order they were allocated, and $Y = \{y_1, y_2, \ldots, y_k\}$ be $i$'s bundle in Round-Robin$({\succ}_i, \bm{\succ}^\mathrm{b}_{-i})$, where goods are indexed by Algorithm \ref{mech:Devnam}. Then, for every $j \in [k]$, we have  $v_i(x_{j} \,|\, X^{j-1}_-) \ge v_i(y_{j} \,|\, X^{j-1}_-)$.
	\end{lemma}

	\begin{proof}
		The way goods in $X$ are indexed, we have that $x_j$ is the good allocated to agent $i$ in round $j$ of Round-Robin$(\bm{\succ}^\mathrm{b})$. Suppose, towards a contradiction, that there is some ${\ell} \in [k]$, for which we have $v_i(y_{\ell} \,|\, X^{\ell-1}_-) > v_i(x_{\ell} \,|\, X^{\ell-1}_-)$. First notice that ${\ell} \neq 1$, as $x_1$ is, by the definition of the bluff profile, a singleton of maximum value for agent $i$ excluding the goods allocated to agents $1$ through $i-1$ in round $1$, regardless of agent $i$'s bid. Thus, ${\ell} \geq 2$.
		
		Let $B \subseteq M$ and $D \subseteq M$ be the sets of goods allocated (to any agent) up to right before a good is allocated to agent $i$ in round $\ell$ in Round-Robin$(\bm{\succ}^\mathrm{b})$ and Round-Robin$({\succ}_i, \bm{\succ}^\mathrm{b}_{-i})$, respectively. Clearly, $|B| = |D| = (\ell -1) n + i -1$. In fact, we claim that in this case the two sets are equal. 
		
		\begin{claim}\label{cl:SetA-submod}
			It holds that $B = D$. Moreover, $\{y_1, \ldots, y_{\ell}\} \subseteq B$.
		\end{claim}

		\begin{proof}[Proof of the claim]
			\renewcommand\qedsymbol{{\small $\boxdot$}}
			We first observe that $v_i(y_{j} \,|\, X^{\ell-1}_-) \ge v_i(y_{\ell} \,|\, X^{\ell-1}_-) > v_i(x_{\ell} \,|\, X^{\ell-1}_-)$, for every $j \in [\ell-1]$, where the first inequality follows from way Algorithm \ref{mech:Devnam} ordered the elements of $Y$. 
			Now consider the execution of Round-Robin$(\bm{\succ}^\mathrm{b})$. Since $x_{\ell}$ was the good allocated to agent $i$ in round $\ell$, $x_{\ell}$ had maximum marginal value for agent $i$ with respect to $X^{\ell-1}_-$ among the available goods. Thus, none of the goods $y_1, \ldots, y_{\ell}$ were available at the time. 
			That is, $y_1, \ldots, y_{\ell}$ were all already allocated to some of the agents (possibly including agent $i$ herself). We conclude that $\{y_1, \ldots, y_{l}\} \subseteq B$.

			Now suppose for a contradiction that $D \neq B$ and consider the execution of Round-Robin$({\succ}_i, \bm{\succ}^\mathrm{b}_{-i})$. Recall that the goods in $B$ are still the $(\ell-1) n+i-1$ most preferable goods for every agent in $N\setminus \{i\}$ according to the profile $({\succ}_i, \bm{\succ}^\mathrm{b}_{-i})$. Therefore, all agents in $N\setminus \{i\}$ will get goods from $B$ allocated to them up to the point when a good is allocated to agent $i$ in round $\ell$, regardless of what ${\succ}_i$ is. If agent $i$ also got only goods from $B$ allocated to her in the first $\ell-1$ rounds of Round-Robin$({\succ}_i, \bm{\succ}^\mathrm{b}_{-i})$, then $D$ would be equal to $B$. Thus, at least one good which is not in $B$ (and thus, not in $\{y_1, \ldots, y_{\ell}\}$) must have been allocated to agent $i$ in the first $\ell-1$ rounds. As a result, at the end of round $\ell-1$, there are at least two goods in $\{y_1, \ldots, y_{\ell}\}$ that have not yet been allocated to $i$. 
			
			However, we claim that up to right before a good is allocated to agent $i$ in round $\ell + 1$, all goods in $B$ (and thus in $\{y_1, \ldots, y_{\ell}\}$ as well) will have been allocated, leaving $i$ with at most $\ell -1$ goods from $\{y_1, \ldots, y_{\ell}\}$ in her final bundle and leading to a contradiction. Indeed, this follows from a simple counting argument. Right before a good is allocated to agent $i$ in round $\ell + 1$, the goods allocated to agents in $N\setminus \{i\}$ are exactly $\ell (n-1) +i-1 \ge (\ell-1) n+i-1 = |B|$. As noted above, agents in $N\setminus \{i\}$ will get goods from $B$ allocated to them as long as they are available. Thus, no goods from $B$, or from $\{y_1, \ldots, y_{\ell}\}$ in particular, remain unallocated right before a good is allocated to agent $i$ in round $\ell + 1$. Therefore, agent $i$ may get at most $\ell-1$ goods from $\{y_1, \ldots, y_{\ell}\}$ (at most $\ell-2$ in the first $\ell-1$ rounds and one in round $\ell$), contradicting the definition of the set $Y$. We conclude that $D = B$.
		\end{proof}
		
		Given the claim, it is now easy to complete the proof. Clearly, in the first $\ell-1$ rounds of Round-Robin$({\succ}_i, \bm{\succ}^\mathrm{b}_{-i})$ at most $\ell-1$ goods from $\{y_1, \ldots, y_{\ell}\}$ have been allocated to agent $i$. However, when it is $i$'s turn in round $\ell$, only goods in $M\setminus D$ are available, by the definition of $D$. By Claim \ref{cl:SetA-submod}, we have $\{y_1, \ldots, y_{l}\} \subseteq D$, and thus there is at least one good $\{y_1, \ldots, y_{\ell}\}$ that is allocated to another agent, which contradicts the definition of $Y$.
	\end{proof}

	We are now ready to state and prove the main result of this section. 
	
	\begin{theorem}\label{thm:RReq}
		When all agents have submodular valuation functions, the bluff profile is a $\frac{1}{2}$-approxi\-mate PNE of Mechanism \ref{alg:MRR}.
		Moreover, this is tight, i.e., for any $\varepsilon >0$, there are instances where the bluff profile is not a $\big(\frac{1}{2} + \varepsilon \big)$-approximate PNE.
	\end{theorem}

	\begin{proof}
		We are going to use the notation used so far in the section and consider the possible deviation of an arbitrary agent $i$. Like in the statement of Lemma \ref{lem:DecMansub}, $X=\{x_1, \ldots,x_k\}$ is agent $i$'s bundle in Round-Robin$(\bm{\succ}^\mathrm{b})$, with goods  indexed in the order they were allocated, and $Y = \{y_1, y_2, \ldots, y_k\}$ is $i$'s bundle in Round-Robin$({\succ}_i, \bm{\succ}^\mathrm{b}_{-i})$, with goods  indexed by Algorithm \ref{mech:Devnam}.
		Also, recall that $X^j_- = \{x_1, \ldots, x_{j}\}$ and $X^j_+ = \{x_j,  \ldots, x_{k}\}$ (and similarly for $Y^j_-$ and $Y^j_+$). We also use the convention that $Y_+^{k+1} = \emptyset$. For any $j\in [k]$, we have
		\begin{align*}
			v_i(X_-^j) - v_i(X_-^{j-1})  &= v_i (x_j \,|\, X_-^{j-1}) \\[0.5ex] 
			&\geq v_i (y_j \,|\, X_-^{j-1}) \\[0.5ex]
			&\geq v_i (y_j \,|\, X_-^{j-1} \cup Y_+^{j+1}) \\[0.5ex]
			&= v_i (X_-^{j-1} \cup Y_+^{j+1} \cup \{y_j\})- v_i (X_-^{j-1} \cup Y_+^{j+1})\\[0.5ex]
			&= v_i (X_-^{j-1} \cup Y_+^{j})- v_i (X_-^{j-1} \cup Y_+^{j+1})\\[0.5ex]
			&\geq v_i(X_-^{j-1} \cup Y_+^{j})- v_i (X_-^{j} \cup Y_+^{j+1}) \,.
		\end{align*}
		The first inequality holds because Lemma \ref{lem:DecMansub} applies on $X$ and $Y$, whereas the second inequality holds because of submodularity. Finally, the last inequality holds since $X_-^{j-1} \subseteq X_-^{j}$ and $v_i(\cdot)$ is non-decreasing, for every $i \in N$. Using these inequalities along with a standard expression of the value of a set as a sum of marginals, we have
		\begin{align*}
			v_i(X)      &= v_i(X_-^k) - v_i(X_-^0) \\[0.5ex]
			&= \sum_{j=1}^{k} \left( v_i(X_-^j) - v_i(X_-^{j-1}) \right) \\[0.5ex]
			&\ge \sum_{j=1}^{k} \left( v_i(X_-^{j-1} \cup Y_+^{j})- v_i (X_-^{j} \cup Y_+^{j+1}) \right) \\[0.5ex]
			&= v_i (X_-^0 \cup Y_+^1) - v_i (X_-^k \cup Y_+^{k+1}) \\[0.5ex] 
			&= v_i(Y) - v_i(X) \,.
		\end{align*}
		Thus, we have $v_i(X) \geq \frac12 \cdot v_i(Y)$, and we conclude that $\bm{\succ}^\mathrm{b}$ is a $\frac{1}{2}$-approximate PNE of Mechanism \ref{alg:MRR}.\medskip
		
		To show that the result is tight, consider an example with two agents and five goods.  
		The valuation function of agent $1$ is additive and defined as follows on the singletons:
		\[	v_1(g_1)= 2 \quad
		v_1(g_2)= 1 \quad
		v_1(g_3)= 1-\varepsilon_1 \quad
		v_1(g_2)= 1-\varepsilon_2 \quad
		v_1(g_5)= 1-\varepsilon_3  \,,
		\]
		where $1 \gg \varepsilon_3>\varepsilon_2>\varepsilon_1 > 0$.

		The valuation function of agent $2$ is OXS\footnote{Roughly speaking, OXS functions generalize unit-demand functions. The set of OXS functions is a strict superset of additive functions and a strict subset of submodular functions. See, \citep{LehmannLN06, Leme17}.} and defined by the maximum matchings in the bipartite graph below, e.g., $v_2(\{g_1, g_2\}) = 2 +1 = 3$ and $v_2(\{g_1, g_4, g_5\}) = 2 + 1 - \varepsilon_2 = 3- \varepsilon_2$. 
		
		\begin{figure}[ht]
			\centering
			\begin{tikzpicture}[y=-0.95cm, minimum size=0.6cm]
				\node [circle, draw] (a) at (0,0) {{\small $g_1$}};
				\node [circle, draw] (b) at (0,1) {{\small $g_2$}};
				\node [circle, draw] (c) at (0,2) {{\small $g_3$}};
				\node [circle, draw] (d) at (0,3) {{\small $g_4$}};
				\node [circle, draw] (e) at (0,4) {{\small $g_5$}};
				
				\node [circle, draw] (a2) at (6,0) {};
				\node [circle, draw] (b2) at (6,1) {};
				\node [circle, draw] (c2) at (6,2) {};
				
				\draw (a) -- node[near start, above] {2} (a2);
				\draw (b) -- node[near start, above] {1} (b2);
				\draw (c) -- node[near start, above] {$1 - \varepsilon_1$} (c2);
				
				\draw (d) to[out=0, in=240] node[pos=0.195, above] {$1 - \varepsilon_2$} (b2);
				\draw (e) to[out=0, in=240] node[pos=0.19, above] {$1 - \varepsilon_3$} (b2);
			\end{tikzpicture}
		\end{figure}

		It is not hard to see that the bluff profile for this instance consists of the following declared ordering by both agents: $g_1>g_2>g_3>g_4>g_5$. The allocation produced by Mechanism \ref{alg:MRR} for the bluff profile is then $A=(A_1,A_2)$, where $A_1= \{g_1,g_3,g_5\}$, and $A_2=\{g_2, g_4\}$. Observe that $v_1(A_1)=4-\varepsilon_1-\varepsilon_3$ and $v_2(A_2) = 1$.
		It is easy to see that there is no profitable deviation for agent $1$, while the maximum value that agent $2$ can attain by deviating is $2-\varepsilon_1-\varepsilon_2$. Agent $2$ achieves this by reporting the preference ranking: $g_3>g_4>g_1>g_2>g_5$ and getting goods $\{g_3, g_4\}$. This implies that for any $\varepsilon >0$ one can chose appropriately small $\varepsilon_1, \varepsilon_2, \varepsilon_3$ so that the bluff profile is not a $\big(\frac{1}{2} + \varepsilon \big)$-approximate PNE.
	\end{proof}

	In Section \ref{sec:b}, we show that every approximate PNE of Mechanism \ref{alg:MRR} results in an approximately \efo allocation. Here, as a warm-up, we start this endeavor with an easy result which holds specifically for the bluff profile (and can be extended to approximate PNE where all agents submit the same preference ranking) but shows a better fairness guarantee than our general Theorem \ref{thm:submod_ef1_n}.

	\begin{theorem}\label{thm:RRguasub}
		When all agents have submodular valuation functions $v_1, \ldots, v_n$, the allocation returned by Round-Robin$(\bm{\succ}^\mathrm{b})$ is $\frac{1}{2}$-\efo  with respect to $v_1, \ldots, v_n$.
		Moreover, this is tight, i.e., for any $\varepsilon >0$, there are instances where this allocation is not $\big(\frac{1}{2} + \varepsilon \big)$-\efo.
	\end{theorem}

	\begin{proof}
		In order to obtain a contradiction, suppose that the allocation $(A_1^\mathrm{b}, A_2^\mathrm{b}, \ldots, A_n^\mathrm{b})$ returned by Round-Robin$(\bm{\succ}^\mathrm{b})$ is not $\frac{1}{2}$-\efo. That is, there exist agents $i$ and $j$ such that $v_i(A_i^\mathrm{b}) < 0.5 \cdot v_i(A_j^\mathrm{b} \setminus \{g\})$, for all $g \in A_j^\mathrm{b}$. 
		We are going to show that this allows us to construct a deviation for agent $i$ where she gets value more than $2 v_i(A_i^\mathrm{b})$, contradicting the fact that $\bm{\succ}^\mathrm{b}$ is a $\frac{1}{2}$-approximate PNE. Recall that using the renaming $h_1, h_2, \ldots$ produced by Algorithm \ref{alg:Grenam}, we have  $A_i^\mathrm{b}= \{ h_i, h_{n+i},\dots, h_{(k-1)n+i}\}$ and $A_j^\mathrm{b}= \{ h_j, h_{n+j},\dots, h_{(k-1)n+j}\}$.
		
		Let $\delta$ be the indicator variable of the event $j<i$, i.e., $\delta$ is $1$ if $j<i$ and $0$ otherwise. We will show that it is possible for agent $i$ to get the  set $\{ h_{\delta n+j}, h_{(1+\delta)n+j}, h_{(2 + \delta)n+j}, \dots, h_{(k-1)n+j}\}$, which is either the entire $A_j^\mathrm{b}$ (when $i<j$) or $A_j^\mathrm{b} \setminus \{h_j\}$ (when $j<i$). In particular, let $\succ_i$ be a preference ranking that starts with all goods in $A_j^\mathrm{b}$ in the same order as they were allocated to agent $j$ in Round-Robin$(\bm{\succ}^\mathrm{b})$, followed by everything else, in any order.
		
		Consider the execution of Round-Robin$(\succ_i, \bm{\succ}_{-i}^\mathrm{b})$. The crucial, yet simple, observation (that makes an inductive argument work) is that the first $i-1$ goods $h_1, \ldots, h_{i-1}$ are allocated as before, then good $h_{\delta n+j}$ (rather than $h_i$) is allocated to agent $i$, and after that the $n-1$ top goods for all agents in $N\setminus \{i\}$ according to $\bm{\succ}_{-i}^\mathrm{b}$ are 
		$h_i, h_{i+1},\dots, h_{\delta n+j-1}, h_{\delta n+j +1},  \dots, h_{n+i-1}$,
		and these are allocated in the next $n-1$ steps of the algorithm. As a result, right before a second good is allocated to agent $i$, the available goods are $h_{n+i}, h_{n+i+1}, \dots, h_m$ exactly as in the execution of Round-Robin$(\bm{\succ}^\mathrm{b})$. 
		
		More generally, right before an $r$-th good is allocated to $i$, her bundle is $\{ h_{\delta n+j}, h_{(1+\delta)n+j}, h_{(2 + \delta)n+j}, \allowbreak \dots, h_{(r-2 + \delta)n+j}\}$, and the available goods are $h_{(r-1)n+i}, h_{(r-1)n+i+1}, \dots, h_m$ (as they were in the execution of Round-Robin$(\bm{\succ}^\mathrm{b})$). Then good $h_{(r-1 + \delta)n+j}$ (rather than $h_{(r-1)n+i}$) is allocated to agent $i$, and after that the $n-1$ top goods for all agents according to $\bm{\succ}_{-i}^\mathrm{b}$ are 
		\[h_{(r-1)n+i}, h_{(r-1)n+i+1},\dots, h_{(r-1 + \delta)n+j - 1}, h_{(r-1 + \delta)n+j + 1},  \dots, h_{rn+i-1} \,,\]
		and they are allocated in the next $n-1$ steps of the algorithm. At the end, agent $i$ gets the entire $A_j^\mathrm{b}$ or $A_j^\mathrm{b} \setminus \{h_j\}$ plus some arbitrary good, depending on whether $i<j$ or $j<i$. In either case, by monotonicity, agent $i$'s value for her bundle is at least $v_i(A_j^\mathrm{b} \setminus \{h_j\}) > 2 v_i(A_i^\mathrm{b})$, where the last inequality follows from our assumption that $(A_1^\mathrm{b}, A_2^\mathrm{b}, \ldots, A_n^\mathrm{b})$ is not $\frac{1}{2}$-\efo. Therefore, by deviating from $\succ^\mathrm{b}$ to $\succ_i$, agent $i$ increases her value by a factor strictly grater than $2$, contradicting Theorem \ref{thm:RReq}.

		\medskip
		
		To show that this factor is tight, we again turn to the example given within the proof of Theorem \ref{thm:RReq}.  
		Recall the allocation produced by Mechanism \ref{alg:MRR} for the bluff profile is $A=(A_1,A_2)$, with $A_1= \{g_1,g_3,g_5\}$ and $A_2=\{g_2, g_4\}$. Observe that agent $1$ is envy-free towards agent $2$ as $v_1(A_1)=4-\varepsilon_1-\varepsilon_3>2-\varepsilon_2=v_1(A_2)$. On the other hand, $v_2(A_2) = 1$, whereas $v_2(A_1) = 4-\varepsilon_1-\varepsilon_3$ and $v_2(A_1 \setminus \{g_1\})=2-\varepsilon_1-\varepsilon_3$. The latter implies that for any $\varepsilon >0$ one can chose appropriately small $\varepsilon_1, \varepsilon_2, \varepsilon_3$ so that the bluff profile does not result in a $\big(\frac{1}{2} + \varepsilon \big)$-\efo allocation with respect to the true valuation functions of the agents. 
	\end{proof}

	\section{Fairness properties of PNE}
	\label{sec:b}
	
	In Section \ref{subsec:RR}, Proposition \ref{prop:RR_ef1_canc}, we state the fairness guarantees of Round-Robin---viewed as an algorithm---when all agents have cancelable valuation functions. 
	So far, we have not discussed this matter for the submodular case. 
	It is not hard to see, however, that Theorem \ref{thm:RRguasub} and the definition of the bluff profile via Algorithm \ref{alg:Grenam} imply that when we have (value oracles for) the valuation functions, then we can use Round-Robin to algorithmically produce $\frac{1}{2}$-\efo allocations. 
	Using similar arguments, we show next that for any preference profile $\bm{\succ} \,= (\succ_1, \ldots, \succ_n)$ and any  $i\in N$, there is always a response $\succ'_i$ of agent $i$ to $\bm{\succ}_{-i}$, such that the allocation returned by Round-Robin$(\succ'_i, \bm{\succ}_{-i})$ is $\frac{1}{2}$-\efo \emph{from agent $i$'s perspective}. 
	
	Towards this, we first need a variant of Algorithm \ref{alg:Grenam} that considers everyone in $N\setminus \{i\}$ fixed to their report in $\bm{\succ}_{-i}$ and greedily determines a ``good'' response for agent $i$. An intuitive interpretation of what Algorithm \ref{alg:Greedy_for_i} below is doing, can be given if one sees Mechanism \ref{alg:MRR} as a sequential game. Then, given that everyone else stays consistent with $\bm{\succ}_{-i}$, agent $i$ \emph{picks} a good of maximum marginal value every time her turn is up.

	\begin{algorithm}[ht]
		\caption{Greedy response of agent $i$ to $\bm{\succ}_{-i}$ 
			\\ Input: $N$, $M$, $\bm{\succ}_{-i}$, value oracle for $v_i$
		}\label{alg:Greedy_for_i}
		\begin{algorithmic}[1]
			\vspace{2pt}\State $S=M$; $X = \emptyset$
			\For{$j = 1, \dots, m$}
			\State $\ell = (j-1) \!\pmod n + 1$
			\If {$\ell = i$}
			\State $x_{\lceil j/n \rceil} = \displaystyle \argmax_{g \in S} v_i(g \,|\, X)$ \Comment{Ties are broken lexicographically.} 
			\State $X=X\cup \{x_{\lceil j/n \rceil}\}$
			\vspace{1pt}\State $S=S\setminus \{x_{\lceil j/n \rceil}\}$
			\Else
			\State $g = \mathrm{top}(\succ_{\ell}, S)$
			\State $S=S\setminus \{g\}$
			\EndIf
			\EndFor \vspace{-1pt}
			\State \Return $x_1 \succ'_i x_2 \succ'_i \ldots \succ'_i x_k \succ'_i \ldots$ \Comment{Arbitrarily complete $\succ'_i$ with goods in $M\setminus X$.}
		\end{algorithmic}
	\end{algorithm}
	
	Proving the next lemma closely follows the proof of Theorem \ref{thm:RReq} but without the need of an analog of Lemma \ref{lem:DecMansub}, as we get this for free from the way the greedy preference profile $\succ'_i$ is constructed. 
	
	\begin{lemma}\label{lem:RR_ef1_sub}
		Assume that agent $i$ has a submodular valuation function $v_i$. If\, $\succ'_i$ is the ranking returned by Algorithm \ref{alg:Greedy_for_i} when given $N$, $M$, $\bm{\succ}_{-i}$, $v_i$, then the allocation $(A'_1, A'_2, \ldots, A'_n)$ returned by Round-Robin$(\succ'_i, \bm{\succ}_{-i})$ is such that 
		for every $j\in N$, with $A'_j\neq\emptyset$, there exists a good $g\in A'_j$, so that $v_i(A'_i) \geq \frac{1}{2} \cdot v_i(A'_j\setminus \{g\})$.
	\end{lemma}

	\begin{proof}
		First, it is straightforward to see that $A'_i = X$, as computed in Algorithm \ref{alg:Greedy_for_i}. Indeed, Algorithm \ref{alg:Greedy_for_i} simulates Mechanism \ref{alg:MRR} for all $j \in N\setminus\{i\}$ and iteratively builds $\succ'_i$, so that in every turn of Round-Robin$(\succ'_i, \bm{\succ}_{-i})$ the good allocated to agent $i$ is one of maximum marginal value. As a result, the goods in $A'_i = X =  \{x_1, x_2, \ldots, x_k\}$ are already indexed in the order they are allocated. 
		
		Now consider an arbitrary $j \in N\setminus\{i\}$ and let $A'_j = Y =  \{y_1, y_2, \ldots, y_k\}$, where goods are again indexed in the order they are allocated in Round-Robin$(\succ'_i, \bm{\succ}_{-i})$. 
		Notice that when good $x_r$ is allocated to agent $i$ in round $r$, goods $y_{r+1}, y_{r+2}, \ldots$ are still available and, by construction of $X$, their marginal value with respect to the set $\{x_1, x_2, \ldots, x_{r-1}\}$ is no better than the marginal value of $x_r$. In particular, $v_i(x_r \,|\, \{x_1, \ldots, x_{r-1}\}) \ge v_i(y_{r+1} \,|\, \{x_1, \ldots, x_{r-1}\})$.
		
		Also, recall the use of $X^r_-$, $X^r_+$, $Y^r_-$, $Y^r_+$ notation from the proof of Theorem \ref{thm:RReq}. We will use a similar calculation here as well, but we will omit the first element of $Y$. For any $r\in [k]$, we have
		
		\begin{align*}
			v_i(X_-^r) - v_i(X_-^{r-1})  &= v_i (x_r \,|\, X_-^{r-1}) \\[0.5ex] 
			&\geq v_i (y_{r+1} \,|\, X_-^{r-1}) \\[0.5ex]
			&\geq v_i (y_{r+1} \,|\, X_-^{r-1} \cup Y_+^{r+2}) \\[0.5ex]
			&= v_i (X_-^{r-1} \cup Y_+^{r+2} \cup \{y_{r+1}\})- v_i (X_-^{r-1} \cup Y_+^{r+2})\\[0.5ex]
			&= v_i (X_-^{r-1} \cup Y_+^{r+1})- v_i (X_-^{r-1} \cup Y_+^{r+2})\\[0.5ex]
			&\geq v_i(X_-^{r-1} \cup Y_+^{r+1})- v_i (X_-^{r} \cup Y_+^{r+2}) \,,
		\end{align*}
		where we used the convention that $Y_+^{k+1} =Y_+^{k+2} = \emptyset$.
		The first inequality holds by the construction of $X$ as discussed above, the second inequality follows from submodularity, and the last inequality holds because $v_i(\cdot)$ is non-decreasing. Using these inequalities and a standard expression of the value of a set as a sum of marginals, we have
		\begin{align*}
			v_i(X)      &= v_i(X_-^k) - v_i(X_-^0) \\[0.5ex]
			&= \sum_{r=1}^{k} \left( v_i(X_-^r) - v_i(X_-^{r-1}) \right) \\[0.5ex]
			&\ge \sum_{r=1}^{k} \left( v_i(X_-^{r-1} \cup Y_+^{r+1})- v_i (X_-^{r} \cup Y_+^{r+2}) \right) \\[0.5ex]
			&= v_i (X_-^0 \cup Y_+^2) - v_i (X_-^k \cup Y_+^{k+2}) \\[0.5ex] 
			&= v_i(Y\setminus \{y_1\}) - v_i(X) \,.
		\end{align*}
		Thus, we have $v_i(A'_i) = v_i(X) \geq \frac12 \cdot v_i(Y\setminus \{y_1\}) = \frac12 \cdot v_i(A'_j\setminus \{y_1\})$. 
	\end{proof}

	\subsection{The Case of Two Agents}
	\label{subsec:two_agents}
	As a warm-up, we begin with the easier case of $n = 2$. Not only the proofs of our main results for submodular and additive functions are much simpler here, but the fairness guarantees are stronger as well.

	\begin{theorem}\label{thm:subeq}
		Let $\alpha \in (0, 1]$. Assume we have a fair division instance with two agents, whose valuation functions $v_1, v_2$ are submodular. Then any allocation that corresponds to a $\alpha$-approximate PNE of the Round-Robin mechanism is $\frac{\alpha}{2}$-\efo with respect to $v_1, v_2$.
	\end{theorem}
	
	\begin{proof}
		Let $\bm{\succ} \,= (\succ_1, \succ_2)$ be a $\alpha$-approximate PNE of Mechanism \ref{alg:MRR} for a given instance, and let $(A_1, A_2)$ be the allocation returned by Round-Robin$(\bm{\succ})$. Consider one of the two agents; we call this agent $i\in [2]$ and the other agent $j$. 
		We are going to show that $v_i(A_i) \geq \frac{\alpha}{2} \cdot v_i(A_j \setminus \{g\})$ for some good $g\in A_j$. 
		
		Suppose that agent $i$ deviates to $\succ'_i$ produced by Algorithm \ref{alg:Greedy_for_i} when given $\bm{\succ}_{-i} \, = (\succ_j)$ and $v_i$, and let $(A'_1, A'_2)$ be the allocation returned by Round-Robin$(\succ'_i, \bm{\succ}_{-i})$. Let $A'_i = \{x_1, x_2, \ldots, x_k\}$ and $A_j \setminus A'_i=  \{y_{t_1}, y_{t_2}, \ldots, y_{t_{\ell}}\}$, where in both sets goods are indexed by the round in which they were allocated in the run of Round-Robin$(\succ'_i, \bm{\succ}_{-i})$. Note that all indices in $A_j \setminus A'_i$ are distinct exactly because $n=2$ and, thus, all these goods are allocated to agent $j$.
		This indexing guarantees that when $x_{t_{\lambda}-1}$ gets allocated, $y_{t_{\lambda}}$ is still available  for $2\le \lambda \le \ell$ and, thus, 
		\begin{equation}\label{eq:greedy_subseq}
			v(x_{t_{\lambda}-1}\,|\,  \{x_1, x_2, \ldots, x_{t_{\lambda} - 2}\}) \ge v(y_{t_{\lambda}}\,|\,  \{x_1, x_2, \ldots, x_{t_{\lambda} - 2}\})\,,
		\end{equation}
		by the way $\succ'_i$ is constructed (see also the proof of Lemma \ref{lem:RR_ef1_sub}).
		Using Theorem \ref{thm:SM}, we have
		\begin{align*}
			v_i(A_j \setminus \{y_{t_1}\})    &\leq   v_i(A'_i) + \!\!\sum_{g \in (A_j \setminus \{y_{t_1}\}) \setminus A'_i} \!\!\!\!\! v(g\,|\,A'_i)\\[0.5ex] 
			& =   v_i(A'_i) + \sum_{\lambda = 2}^{\ell} v(y_{t_{\lambda}}\,|\,A'_i)\\[0.5ex]
			&\le v_i(A'_i) + \sum_{\lambda = 2}^{\ell} v(y_{t_{\lambda}}\,|\,  \{x_1, x_2, \ldots, x_{t_{\lambda} - 2}\})\\[0.5ex]
			&\le  v_i(A'_i) + \sum_{\lambda = 2}^{\ell} v(x_{t_{\lambda}-1}\,|\,  \{x_1, x_2, \ldots, x_{t_{\lambda} - 2}\})\\[0.5ex]
			&\le  v_i(A'_i) + \sum_{\lambda = 1}^{k} v(x_{\lambda}\,|\,  \{x_1, x_2, \ldots, x_{\lambda - 1}\})\\[0.5ex]
			&=  v_i(A'_i) + v_i(A'_i)\\[0.5ex]
			&\le  \frac{2}{\alpha} \cdot v_i(A_i) \,,
		\end{align*}
		where the first inequality follows directly from Theorem \ref{thm:SM}, the second one follows from submodularity, the third inequality holds because of \eqref{eq:greedy_subseq}, the fourth one follows from the monotonicity of $v_i$, and the last inequality follows from the fact that $\bm{\succ}$ is a $\alpha$-approximate PNE and thus $v_i(A_i)\ge \alpha\cdot v_i(A'_i)$. We conclude that $(A_1, A_2)$ is $\frac{\alpha}{2}$-\efo with respect to the underlying valuation functions. 
	\end{proof}
	
	For additive valuation functions we can get a slightly stronger fairness guarantee, which we show that is also tight for any $\alpha$, with an even easier proof. Note that this reproduces the result of \citet{AmanatidisBFLLR21} for exact PNE in the case of two agents. 
	
	\begin{theorem}\label{thm:additive_fair_eq}
		Let $\alpha \in (0, 1]$. Assume we have a fair division instance with two agents, whose valuation functions $v_1, v_2$ are additive. Then any allocation that corresponds to a $\alpha$-approximate PNE of the Round-Robin mechanism is $\frac{\alpha}{2-\alpha}$-\efo with respect to $v_1, v_2$. This is tight, i.e., for any $\varepsilon> 0$, there are instances where a $\alpha$-approximate PNE does not correspond to a $(\frac{\alpha}{2-\alpha}+ \varepsilon)$-\efo allocation.
	\end{theorem}
	
	\begin{proof}
		Let $\bm{\succ} \,= (\succ_1, \succ_2)$, $A_1$, $A_2$ be as in the proof of Theorem \ref{thm:subeq}, but now consider the deviation of agent $i$ to $\succ'_i$ which is a strict version of her true preference ranking $\succcurlyeq^*_i$. Again, let $(A'_1, A'_2)$ be the allocation returned by Round-Robin$(\succ'_i, \bm{\succ}_{-i})$.
		
		Let $g$ be  good of maximum value in $A'_j$ according to $v_i$. Since $\succ'_i$ is a true preference ranking of agent $i$, according to Proposition \ref{prop:RR_ef1_canc} $(A'_1, A'_2)$ is \efo from the point of view of agent $i$. That is, we have $v_i(A'_i)\ge v_i(A'_j \setminus \{g\})$ and, thus, $v_i(A'_i)\ge \frac{1}{2}\cdot v_i(M \setminus \{g\})$. Therefore, 
		\begin{align*}
			v_i(A_j \setminus \{g\})    &=   v_i(M\setminus \{g\}) - v_i(A_i)\\[0.5ex]
			& \le   2\cdot v_i(A'_i) - v_i(A_i)\\[0.5ex]
			&\le  \frac{2}{\alpha} \cdot v_i(A_i) - v_i(A_i)\\[0.5ex]
			&=  \frac{2 - \alpha}{\alpha} \cdot v_i(A_i) \,,
		\end{align*}
		where the second inequality follows from the fact that $\bm{\succ}$ is a $\alpha$-approximate PNE and thus $v_i(A_i)\ge \alpha\cdot v_i(A'_i)$. We conclude that $(A_1, A_2)$ is $\frac{\alpha}{2  - \alpha}$-\efo with respect to $v_1, v_2$. 
		\medskip

		To see that this guarantee is tight, consider an instance with two agents, and a set of five goods $\{g_1, g_2, \ldots, g_5\}$. In addition, let the valuation functions of the agents to be additive and defined by:
		
		\begin{minipage}{0.5\textwidth}
			\[
			v_1(g_j)=  
			\begin{cases}
				6, & \text{if $j=1$} \\
				3 + \delta, & \text{if $j=2$}  \\
				3, & \text{if $j=3$} \\
				0.5 + \delta, & \text{if $j=4$} \\
				0.5, & \text{if $j=5$}
			\end{cases} 
			\]
		\end{minipage}
		\begin{minipage}{0.31\textwidth}
			\[
			v_2(g_j)=  
			\begin{cases}
				6  \beta, & \text{if $j=1$} \\
				3 \beta + \delta, & \text{if $j=2$}  \\
				3 \beta, & \text{if $j=3$} \\
				0.5 + \delta, & \text{if $j=4$} \\
				0.5, & \text{if $j=5$}
			\end{cases} 
			\]
		\end{minipage}\smallskip
		\\
		where $0.5\gg \delta$, and $\beta>\frac{1}{6}+ \delta$. Now suppose that the agents bid as follows: Agent $1$ bids truthfully (i.e., an ordering $\succ_1$ that is consistent with her true valuation function), while agent $2$ bids $g_5 \succ_2 g_4 \succ_2 g_1 \succ_2 g_2 \succ_2 g_3$. It is easy to confirm that the produced allocation is $A=(A_1, A_2)=(\{g_1, g_2, g_3\}, \{g_4, g_5\})$. Regarding agent 1, she takes her three most desirable goods in this allocation so there is no profitable deviation for her. For the same reason, she is envy-free towards agent 2.
		
		Moving to agent 2, by observing her valuation function, we immediately derive that she is  $\frac{1 +\delta}{6 \beta + \delta}$-\efo towards agent 1. The only thing that remains, is to check how much agent $2$ can improve her utility through deviating. Initially notice that agent $2$ cannot get good $g_1$ regardless of her bid as this good is taken by agent $1$ in round 1. At the same time, it is easy to verify that she cannot get both goods $g_2$ and $g_3$ due to the declared ordering of agent 1. Thus, the best bundle of goods that she can acquire is $\{g_2, g_4\}$ by deviating to the bid: $g_2 \succ'_2 g_4 \succ'_2 g_1 \succ'_2 g_3 \succ'_2 g_5$ and attain a value of $3 \beta + 0.5 +2 \delta$. 
		
		By setting $\alpha = \frac{1+\delta}{3 \beta + 0.5 +2 \delta}$ we trivially have that $(\succ_1, \succ_2)$ is a $\alpha$-approximate PNE. On the other hand, for a given $\varepsilon > 0$, we have $\frac{\alpha}{2-\alpha}+ \varepsilon = \frac{1+\delta}{6 \beta + 3 \delta} + \varepsilon$ which is strictly larger than $\frac{1 +\delta}{6 \beta + \delta}$ for sufficiently small $\delta$. That is, there is a choice of $\delta$ so that the $\alpha$-approximate PNE $(\succ_1, \succ_2)$ is not $\frac{\alpha}{2-\alpha}+ \varepsilon$-\efo.
	\end{proof}

	\subsection{The Case of \texorpdfstring{$n$}{n} Agents}
	\label{subsec:n_agents}
	
	Looking back at the proofs of Theorems \ref{thm:subeq} and \ref{thm:additive_fair_eq}, the obvious fact that everything not in $A_i$ or $A'_i$ was allocated to agent $j$ played a key role in proving our sharp bounds. Moving to the general case of $n$ agents, there is no reason to expect that we have some control on how the goods are redistributed between agents in $N\setminus \{i\}$ when agent $i$ deviates from an (approximate) equilibrium. Surprisingly, we show that this redistribution does not favor any agent too much from $i$'s perspective when the valuation functions are submodular  or subadditive cancelable (Lemmata \ref{lem:cardinality} and \ref{lem:twoj}).
	Consequently, the main results of this section have similar flavor not only with respect to their statements, but with respect to their proofs as well. 
	
	\begin{theorem}\label{thm:submod_ef1_n}
		Let $\alpha \in (0, 1]$. For instances with submodular valuation functions $\{v_i\}_{i\in N}$, any $\alpha$-approximate PNE of the Round-Robin mechanism is $\frac{\alpha}{3}$-\efo with respect to $\{v_i\}_{i\in N}$.
	\end{theorem}

	\begin{theorem}\label{thm:cancelable_ef1_n}
		Let $\alpha \in (0, 1]$. For instances with subadditive cancelable valuation functions $\{v_i\}_{i\in N}$, any $\alpha$-approximate PNE of the Round-Robin mechanism is $\frac{\alpha}{2}$-\efo with respect to $\{v_i\}_{i\in N}$.
	\end{theorem}
	
	As the proofs of both theorems have the same general structure and share Lemmata \ref{lem:cardinality} and \ref{lem:twoj}, we begin with some common wording and notation, consistent with our proofs for two agents.
	Given any instance, we use $\bm{\succ} \,= (\succ_1, \ldots, \succ_n)$ for an arbitrary $\alpha$-approximate PNE of Mechanism \ref{alg:MRR}. We then consider the deviation of some agent $i$ to a preference ranking $\succ'_i$; in the submodular case $\succ'_i$ is the output of Algorithm \ref{alg:Greedy_for_i} when given $\bm{\succ}_{-i}$ and $v_i$, whereas in the cancelable case $\succ'_i$ is a strict version of $i$'s true preference ranking $\succcurlyeq^*_i$. We use $(A_1, \ldots, A_n)$ and $(A'_1, \ldots, A'_n)$ to denote the allocations returned by Round-Robin$(\bm{\succ})$ and Round-Robin$(\succ'_i, \bm{\succ}_{-i})$, respectively.

	In order to show that $(A_1, \ldots, A_n)$ as $\frac{\alpha}{\kappa}$-\efo from agent $i$'s perspective (where $\kappa$ is $3$ for submodular and $2$ for cancelable functions), we use the stronger \efo guarantees that $(A'_1, \ldots, A'_n)$ has from her perspective. To this end, we use $h_r^{\ell}$ to denote the good that was allocated to an agent $\ell\in N$ in round $r$ of Round-Robin$(\succ'_i, \bm{\succ}_{-i})$. In particular, $A'_i = \{h^i_1, h^i_2, \ldots, h^i_k\}$; recall that $k=m/n$. Further, given that we have fixed agent $i$, we use $S_r$ and $S'_r$, for $0\le r\le k-1$, to denote the set of goods that had been allocated up to right before a good was allocated to $i$ in round $r+1$ of Round-Robin$(\bm{\succ})$ and Round-Robin$(\succ'_i, \bm{\succ}_{-i})$, respectively. 
	That is, for $0\le r\le k-1$, $S_r$ and $S'_r$ contain the goods allocated in steps $1$ through  $rn+i-1$ of Round-Robin$(\bm{\succ})$ and Round-Robin$(\succ'_i, \bm{\succ}_{-i})$, respectively.

	For the next technical lemma we assume that the valuation functions are either submodular or cancelable and, in each case, we use the corresponding $\succ'_i$ as described above.
	
	\begin{lemma}\label{lem:cardinality}
		For any $r\in [k]$, right before an $r$-th good is allocated to agent $i$ in Round-Robin$(\bm{\succ})$, there are at most $r-1$ goods from $S'_{r-1}$ that are still unallocated, i.e.,  $\left|S'_{r-1}\setminus S_{r-1}\right| \leq r-1$.
	\end{lemma}

	\begin{proof}
		We will prove the statement using induction on $r$.
		For $r = 1$, it is straightforward that $S_0 = S'_0$, as the preference  rankings of agents $1$ through $i-1$ are the same in the two runs of the mechanism and, thus, the first goods allocated to them are exactly the same.
		
		Now suppose that the statement is true for every round up to round $r$; we will show that it is true for round $r+1$ as well. 
		Initially, observe that if the number of unallocated goods from $S'_{r-1}$ is $r-1$ right before a good is allocated to agent $i$ in round $r$, it will trivially be at most $r-1$ right before a good is allocated to agent $i$ in round $r+1$ (as the number of unallocated goods from any set cannot increase as the allocation progresses). That is, $\left|S'_{r-1}\setminus S_{r}\right| \leq r-1$. 
		
		Notice that the goods that might cause $S'_{r}\setminus S_{r}$ to increase are the elements of 
		\[S'_{r}\setminus S'_{r-1} = \{h^i_r, h^{i+1}_r, \ldots, h^{n}_r, h^1_{r+1}, h^2_{r+1}, \ldots, h^{i-1}_{r+1}\} \,,\]
		and suppose that there are $\lambda$ goods therein which are still unallocated right before a good is allocated to agent $i$ in round $r+1$ of Round-Robin$(\bm{\succ})$. Clearly, if $\lambda \le 1$, we are done. 
		So, assume that $\lambda \ge 2$. This means that there are $\lambda - 1\ge 1$ unallocated goods in $(S'_{r}\setminus S'_{r-1}) \setminus\{h^i_r\}$. 
		Let $g$ be one of these goods and let $j$ be the agent to  whom $g$ was given, i.e., $g = h^{j}_{\bar{r}}$, where $\bar{r} = r$, if $j>i$, and  $\bar{r} = r+1$, if $j<i$. 
		In either case, notice that according to $\succ_j$ the good $g$ is better than any good in $M\setminus S'_r$ or else it would not have been allocated to $j$ at round $\bar{r}$ of Round-Robin$(\succ'_i, \bm{\succ}_{-i})$ when everything in $M\setminus S'_r$ is still available. 
		We claim that $g$ does not increase the number of elements in $S'_{r}\setminus S_{r}$. Indeed, given that $g$ was available during step $(\bar{r}-1)n+j$ of Round-Robin$(\bm{\succ})$ and that $j$'s declared preference ranking is still $\succ_j$, the only possibility is that during that step one of the unallocated goods from $S'_{r-1} \cup \{h^i_r, h^{i+1}_r, \ldots, h^{j-1}_{\bar{r}}\}$ was allocated to $j$ instead.
		
		Therefore, the only good out of the $\lambda$ candidate goods of $S'_{r}\setminus S'_{r-1}$ which might count towards the number of elements in $S'_{r}\setminus S_{r}$ is $h^i_r$. We conclude that $S'_{r}\setminus S_{r} \le (r-1) +1 = r$.
	\end{proof}

	Lemma \ref{lem:cardinality} is global, illustrating that the sets $S_r$ and $S'_r$ cannot differ in more than a $1/n$-th of their elements.  The next lemma shows that no agent 
	can accumulate too many goods from $S'_r$, for any $0\le r\le k-1$.
	Again, we assume that the valuation functions are either submodular or cancelable and, in each case, the appropriate $\succ'_i$ is used as discussed after the statements of Theorems \ref{thm:subeq} and \ref{thm:additive_fair_eq}. Note that $S'_0$ in the lemma's statement contains exactly these goods which we will exclude when showing the \efo guarantee for our two theorems.

	\begin{lemma}\label{lem:twoj}
		For any $r\in [k]$ and any $j \in N$, agent $j$ gets at most $2 (r-1)$ goods from $S'_{r-1} \setminus S'_0$ in the allocation $(A_1, \ldots, A_n)$ returned by Round-Robin$(\bm{\succ})$, i.e., $|A_j \cap (S'_{r-1}\setminus S'_0)| \le 2 (r-1)$.
	\end{lemma}

	\begin{proof}
		Fix an $r\in [k]$ and a $j\in N$. Consider the end of step $(r-1)n +i -1$ of Round-Robin$(\bm{\succ})$, i.e., right before an $r$-th good is allocated to agent $i$. Ignoring all the goods allocated before $i$ got her first good, agent $j$ has received exactly $r-1$ goods up to this point. As a result, the number of goods allocated to $j$ from $S'_{r-1}\setminus S'_0$ at this point is at most $r-1$. 
		
		At the same time, the number of goods from $S'_{r-1}\setminus S'_0$ that might end up in $A_j$ in any future steps of Round-Robin$(\bm{\succ})$ are at most as many as the goods from $S'_{r-1}$ that are still unallocated at the end of step $(r-1)n +i -1$. The latter, by Lemma \ref{lem:cardinality}, are also at most $r-1$.
		
		From these two observations, we have that the final bundle $A_j$ of agent $j$ may contain at most $2 (r-1)$ goods from $S'_{r-1} \setminus S'_0$.
	\end{proof}
	
	With Lemma \ref{lem:twoj} at hand, we are now ready to prove Theorems \ref{thm:submod_ef1_n} and \ref{thm:cancelable_ef1_n};

	\begin{proof}[Proof of Theorem \ref{thm:submod_ef1_n}]
		We, of course, adopt the notation that has been used throughout this section, focusing on an arbitrary agent $i\in N$ and assuming that her deviation $\succ'_i$ has been the output of Algorithm \ref{alg:Greedy_for_i} with input $\bm{\succ}_{-i}$ and $v_i$. In particular, $(A_1, \ldots, A_n)$ and $(A'_1, \ldots, A'_n)$ are the allocations returned by Round-Robin$(\bm{\succ})$ and Round-Robin$(\succ'_i, \bm{\succ}_{-i})$, respectively.
		
		Consider another agent $j\in N\setminus\{i\}$.
		Let $A'_i = \{x_1, x_2, \ldots, x_k\}$ and $A_j= \{y_{1}, y_{2}, \ldots, y_{k}\}$, where in both sets goods are indexed in the order in which they were allocated in the run of Round-Robin$(\succ'_i, \bm{\succ}_{-i})$. For $A'_i$, this means that $x_{r}$ was allocated in round $r$ for all $r\in [k]$. For $A_j$, this indexing guarantees that 
		for every $0\le \ell < r \le k-1$, the goods in $A_j \cap (S'_{\ell}   \setminus S'_{\ell-1})$ all have smaller indices than the goods in $A_j \cap (S'_{r}   \setminus S'_{r-1})$ (where we use the convention that $S'_{-1} = \emptyset$). We further partition $A_j\setminus \{y_{1}\}$ to
		$Y_1=\{y^1_1,\ldots, y^1_{\tau_1}\}$ and $Y_2= \{y^2_1,\ldots, y^2_{\tau_2}\}$ which contain the goods of $A_j\setminus \{y_{1}\}$  with odd and even indices, respectively, and are both renamed according to Algorithm  \ref{mech:Devnam} with inputs $A'_i$, $Y_{1}$, $v_i$, and $A'_i$, $Y_{2}$, $v_i$, respectively. Clearly, $\tau_1 = \lfloor\frac{k-1}{2}\rfloor$ and $\tau_2 = \lceil\frac{k-1}{2}\rceil$.

		By Lemma \ref{lem:twoj}, we have that $A_j$ contains at most $2 (r-1)$ goods from $S'_{r-1} \setminus S'_0$, for any $r\in [k]$. The original ordering $y_{1}, y_{2}, \ldots$ of the goods in $A_j$ and the way $A_j\setminus \{y_{1}\}$  was partitioned into $Y_{1}$ and $Y_{2}$ imply that 
		$\left| |Y_{1} \cap (S'_{r-1}\setminus S'_0)| - |Y_{2} \cap (S'_{r-1}\setminus S'_0)| \right| \le 1$
		and, thus, each of $Y_{1}$ and $Y_{2}$ contains at most $r-1$ goods from $S'_{r-1} \setminus S'_0$.
		
		We also claim that, for $\ell\in\{1, 2\}$ and $r\in [\tau_{\ell}]$, we have 
		\begin{equation}\label{eq:odd-even_subseq}
			v_i(x_r \,|\, \{x_1, \ldots, x_{r-1}\}) \ge v_i(y^{\ell}_{r} \,|\, \{x_1, \ldots, x_{r-1}\})\,.
		\end{equation}
		Suppose not. That is, there are $\ell\in\{1, 2\}$ and $r\in [\tau_{\ell}]$ so that 
		\eqref{eq:odd-even_subseq} is violated. Note that, by the way Algorithm \ref{mech:Devnam} ordered the elements of $Y_1$ and $Y_2$, this implies
		\[v_i(x_r \,|\, \{x_1, \ldots, x_{r-1}\}) < v_i(y^{\ell}_{r} \,|\, \{x_1, \ldots, x_{r-1}\}) \le v_i(y^{\ell}_{t} \,|\, \{x_1, \ldots, x_{r-1}\})\,,\]
		for all $t \in [r]$. Since $x_{r}$ was the good allocated to agent $i$ at step $(r-1)n +i$ of Round-Robin$(\succ'_i, \bm{\succ}_{-i})$, $x_{r}$ had maximum marginal value for $i$ with respect to $\{x_1, \ldots, x_{r-1}\}$ among the available goods. Thus, none of the goods $y^{\ell}_1, \ldots, y^{\ell}_{r}$ were available at the time, i.e.,  $y^{\ell}_1, \ldots, y^{\ell}_{r} \in S'_{r-1}$. Given that the only good of $A_j$ that could possibly be in $S'_0 = S_0$ was $y_{1}$ which is not in $Y_1 \cup Y_2$. Therefore, $y^{\ell}_1, \ldots, y^{\ell}_{r} \in S'_{r-1}\setminus S'_0$, which contradicts the fact that $|Y_{\ell} \cap (S'_{r-1}\setminus S'_0)| \le r-1$. We conclude that \eqref{eq:odd-even_subseq} holds for all $\ell\in\{1, 2\}$ and $r\in [\tau_{\ell}]$.
		
		We are now ready to apply Theorem \ref{thm:SM} to bound the value of $A_j\setminus \{y_{1}\}$. We have
		\begin{align*}
			v_i(A_j \setminus \{y_{1}\})    &\leq   v_i(A'_i) + \!\!\sum_{g \in (A_j \setminus \{y_{1}\}) \setminus A'_i} \!\!\!\!\! v(g\,|\,A'_i)\\[0.5ex]
			& =   v_i(A'_i) + \!\sum_{g \in Y_1 \setminus A'_i} \!\! v(g\,|\,A'_i) + \!\sum_{g \in Y_2 \setminus A'_i}   \!\! v(g\,|\,A'_i)\\[0.5ex]
			& =  v_i(A'_i) + \sum_{\ell = 1}^{\tau_1} v(y^1_{{\ell}}\,|\,A'_i) + \sum_{\ell = 1}^{\tau_2} v(y^2_{{\ell}}\,|\,A'_i)\\[0.5ex]
			&\le  v_i(A'_i) + \sum_{\ell = 1}^{\tau_1} v(y^1_{{\ell}}\,|\,  \{x_1, \ldots, x_{\ell - 1}\}) + \sum_{\ell = 1}^{\tau_2} v(y^2_{{\ell}}\,|\,  \{x_1, \ldots, x_{\ell - 1}\}) \\[0.5ex]
			&\le  v_i(A'_i) + \sum_{\ell = 1}^{\tau_1} v(x_{{\ell}}\,|\,  \{x_1, \ldots, x_{\ell - 1}\}) + \sum_{\ell = 1}^{\tau_2} v(x_{{\ell}}\,|\,  \{x_1, \ldots, x_{\ell - 1}\}) \\[0.5ex]
			&\le  v_i(A'_i) + 2 \cdot \!\sum_{\ell = 1}^{k} v(x_{\ell}\,|\,  \{x_1, x_2, \ldots, x_{\ell - 1}\})\\[0.5ex]
			&=  v_i(A'_i) + 2 \cdot v_i(A'_i)\\[0.5ex]
			&\le  \frac{3}{\alpha} \cdot v_i(A_i) \,,
		\end{align*}
		where the first inequality follows directly from Theorem \ref{thm:SM}, the second one follows from submodularity, the third inequality holds because of \eqref{eq:odd-even_subseq}, the fourth one follows from the monotonicity of $v_i$, and the last inequality follows from the fact that $\bm{\succ}$ is a $\alpha$-approximate PNE and thus $v_i(A_i)\ge \alpha\cdot v_i(A'_i)$. We conclude that $(A_1, A_2, \ldots, A_n)$ is $\frac{\alpha}{3}$-\efo with respect to the underlying valuation functions. 
	\end{proof}
	
	\begin{proof}[Proof of Theorem \ref{thm:cancelable_ef1_n}] 
		Note that in the proof of Theorem \ref{thm:subeq}, the submodularity of $v_i$ is not used until the final bounding of $A_j\setminus \{y_{1}\}$. Up to that point, the proof here is essentially identical (the only difference being that now $\succ'_i$ is a strict version of $i$'s true preference ranking $\succcurlyeq^*_i$ but this does not change any of the arguments). In particular, for $A'_i = \{x_1, x_2, \ldots, x_k\}$, $A_j= \{y_{1}, y_{2}, \ldots, y_{k}\}$,  
		$Y_1=\{y^1_1,\ldots, y^1_{\tau_1}\}$, and $Y_2= \{y^2_1,\ldots, y^2_{\tau_2}\}$, like in the proof of Theorem \ref{thm:subeq}, we still have \eqref{eq:odd-even_subseq}, for any $\ell\in\{1, 2\}$ and $r\in [\tau_{\ell}]$, i.e., 
		$v_i(x_r \,|\, \{x_1, \ldots, x_{r-1}\}) \ge v_i(y^{\ell}_{r} \,|\, \{x_1, \ldots, x_{r-1}\})$.
		
		Notice that \eqref{eq:odd-even_subseq} can be rewritten as $v_i(\{x_1, \ldots, x_{r-1}, x_r\}) \ge v_i(\{x_1, \ldots, x_{r-1}, y^{\ell}_{r}\})$.
		Since $v_1$ is cancelable, the latter implies that $v_i(x_r) \ge v_i(y^{\ell}_{r})$, for $\ell\in\{1, 2\}$ and $r\in [\tau_{\ell}]$. Now we apply Lemma \ref{lem:cans_ord} to get $v_i(\{x_1, x_2, \ldots, x_{\tau_{\ell}}\}) \ge v_i(Y_{\ell})$, for $\ell\in\{1, 2\}$.
		At this point, we can easily bound the value of $A_j\setminus \{y_{1}\}$. We have
		\begin{align*}
			v_i(A_j \setminus \{y_{1}\})    & =    v_i(Y_1 \cup Y_2)\\[0.5ex]
			& \le    v_i(Y_1) + v_i(Y_2) \\[0.5ex]
			& \le v_i(\{x_1, x_2, \ldots, x_{\tau_{1}}\}) + v_i(\{x_1, x_2, \ldots, x_{\tau_{2}}\})\\[0.5ex]
			&\le  v_i(A'_i) + v_i(A'_i) \\[0.5ex]
			&\le  \frac{2}{\alpha} \cdot v_i(A_i) \,,
		\end{align*}
		where the first inequality follows from subadditivity, the third one follows from the monotonicity of $v_i$, and the last inequality follows from the fact that $\bm{\succ}$ is a $\alpha$-approximate PNE. We conclude that $(A_1, \ldots, A_n)$ is $\frac{\alpha}{2}$-\efo with respect to the underlying valuation functions. 
	\end{proof}

	The ${\alpha}/({2-\alpha})$ upper bound of Theorem \ref{thm:additive_fair_eq} for the additive case applies to both submodular and subadditive cancelable valuation functions, leaving a very small gap for the latter. For the submodular case, we improve this upper bound to ${\alpha}/{2}$.

	\begin{proposition}\label{thm:oxs_lower_bound}
		Let $\alpha, \varepsilon \in (0, 1]$. For instances with submodular valuation functions $\{v_i\}_{i\in N}$, a $\alpha$-approximate PNE of the Round-Robin mechanism may not be $(\frac{\alpha}{2} +\varepsilon)$-\efo with respect to $\{v_i\}_{i\in N}$.
	\end{proposition}

	\begin{proof}
		We construct an instance with four agents and nine goods, i.e., $N=[4]$ and $M=\{g_1, g_2, \ldots, g_9 \}$. Let $1\gg  \varepsilon_1>\varepsilon_2>\varepsilon_3>\varepsilon_4>\varepsilon_5>\varepsilon_6$ and $\beta >({1+\varepsilon_4})/{2}$. The first three agents have additive valuation functions, defined as follows:
		
		\begin{minipage}{0.31\textwidth}
			\[
			v_1(g_j)=  
			\begin{cases}
				5, & \text{if $j=1$} \\
				\varepsilon_5, & \text{if $j=2$}  \\
				\varepsilon_6, & \text{if $j=3$} \\
				1, & \text{if $j=4$} \\
				2, & \text{if $j=5$} \\
				\varepsilon_1, & \text{if $j=6$} \\
				\varepsilon_2, & \text{if $j=7$}  \\
				\varepsilon_3, & \text{if $j=8$} \\
				\varepsilon_4, & \text{if $j=9$} 
			\end{cases} 
			\]
		\end{minipage}
		\begin{minipage}{0.31\textwidth}
			\[
			v_2(g_j)=  
			\begin{cases}
				\varepsilon_5, & \text{if $j=1$} \\
				5, & \text{if $j=2$}  \\
				\varepsilon_6, & \text{if $j=3$} \\
				1, & \text{if $j=4$} \\
				\varepsilon_1, & \text{if $j=5$} \\
				\varepsilon_2, & \text{if $j=6$} \\
				2, & \text{if $j=7$}  \\
				\varepsilon_3, & \text{if $j=8$} \\
				\varepsilon_4, & \text{if $j=9$} 
			\end{cases} 
			\]
		\end{minipage}
		\begin{minipage}{0.31\textwidth}
			\[
			v_3(g_j)=  
			\begin{cases}
				\varepsilon_5, & \text{if $j=1$} \\
				\varepsilon_6, & \text{if $j=2$}  \\
				5, & \text{if $j=3$} \\
				\varepsilon_1, & \text{if $j=4$} \\
				\varepsilon_2, & \text{if $j=5$} \\
				2, & \text{if $j=6$} \\
				\varepsilon_3, & \text{if $j=7$}  \\
				\varepsilon_4, & \text{if $j=8$} \\
				1, & \text{if $j=9$}. 
			\end{cases} 
			\]
		\end{minipage}
		\medskip
		
		Agent $4$ has an OXS (and, thus, submodular) valuation function that is defined by the maximum weight matchings in the bipartite graph below.

		
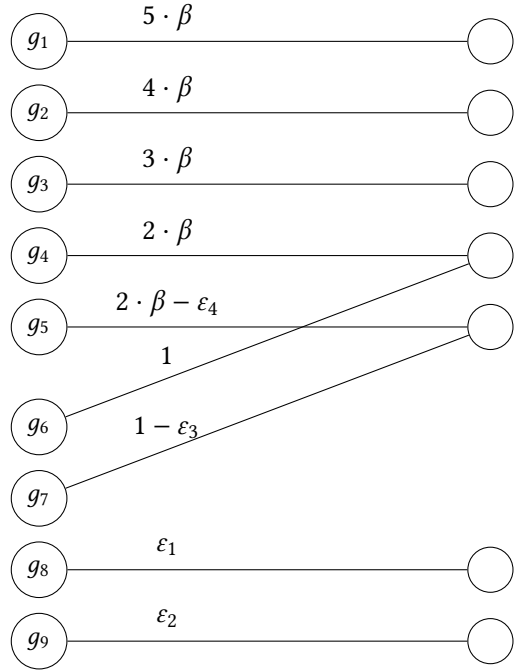
\begin{figure}[ht]
			\centering
			\begin{tikzpicture}[y=-0.95cm, minimum size=0.6cm]
				\node [circle, draw] (a) at (0,0) {\small {$g_1$}};
				\node [circle, draw] (b) at (0,1) {\small {$g_2$}};
				\node [circle, draw] (c) at (0,2) {\small {$g_3$}};
				\node [circle, draw] (d) at (0,3) {\small {$g_4$}};
				\node [circle, draw] (e) at (0,4) {\small {$g_5$}};
				\node [circle, draw] (f) at (0,5.4) {\small {$g_6$}};
				\node [circle, draw] (g) at (0,6.4) {\small {$g_7$}};
				\node [circle, draw] (h) at (0,7.4) {\small {$g_8$}};
				\node [circle, draw] (i) at (0,8.4) {\small {$g_9$}};
				
				\node [circle, draw] (a2) at (6,0) {};
				\node [circle, draw] (b2) at (6,1) {};
				\node [circle, draw] (c2) at (6,2) {};
				\node [circle, draw] (d2) at (6,3) {};
				\node [circle, draw] (e2) at (6,4) {};
				\node [circle, draw] (f2) at (6,7.4) {};
				\node [circle, draw] (g2) at (6,8.4) {};
				
				\draw (a) -- node[near start, above] {$5 \cdot \beta$} (a2);
				\draw (b) -- node[near start, above] {$4 \cdot \beta$} (b2);
				\draw (c) -- node[near start, above] {$3 \cdot \beta$} (c2);
				\draw (d) -- node[near start, above] {$2 \cdot \beta$} (d2);
				\draw (e) -- node[near start, above] {$2 \cdot \beta - \varepsilon_4$} (e2);
				
				\draw (f) -- node[near start, above] {1} (d2);
				\draw (g) -- node[near start, above] {$1 - \varepsilon_3$} (e2);
				
				\draw (h) -- node[near start, above] {$\varepsilon_1$} (f2);
				\draw (i) -- node[near start, above] {$\varepsilon_2$} (g2);
			\end{tikzpicture}
		\end{figure}
		
		Now consider a bidding profile where the first three agents bid truthfully (i.e., they bid the strict preference rankings $\succ^*_1, \succ^*_2, \succ^*_3$ which are consistent with $v_, v_2, v_3$), while the fourth agent bids the preference ranking $\succ_4$:  $g_3 \succ_4 g_6 \succ_4 g_8 \succ_4 g_1 \succ_4 g_2 \succ_4 g_4 \succ_4 g_5 \succ_4 g_7 \succ_4 g_9$. It is easy to confirm that the produced allocation is $(A_1, A_2, A_3, A_4)=\{\{g_1,g_4,g_5\},\{g_2,g_7\},\{g_3, g_9\}, \{g_6,g_8\}\}$.

		We first examine the first three agents. Agents $1$ and $2$ get their most valuable goods in this allocation something that implies that there is no profitable deviation for them. For the same reason they are also envy-free towards the other agents. Regarding agent $3$, the only bundle that improves her utility is
		$\{g_3, g_6\}$. However, there is no bid that she can report and get these two goods. The reason for this is that if she does not get good $g_3$ in round 1 of Mechanism \ref{alg:MRR} (by not declaring it as her best good among the available ones), then $g_3$ is lost to agent $4$. If, on the other hand, she gets good $g_3$ in round 1 (by declaring it as her best good among the available ones), then good $g_6$ is lost to agent $4$. Therefore, there is no profitable deviation for her. Finally, it is easy to see that she is also envy-free towards the other agents.
		
		Moving to agent $4$, we have that
		\[	v_4(A_i)=  
		\begin{cases}
			v_4(g_1) + 4 \beta -\varepsilon_4, & \text{if $i=1$} \\
			v_4(g_2) + 1 -\varepsilon_3, & \text{if $i=2$}  \\
			v_4(g_3) + \varepsilon_2, & \text{if $j=3$} \\
			1 + \varepsilon_1, & \text{if $j=4$}, 
		\end{cases} 
		\]
		where $g_1, g_2, g_3$ are the most valuable goods from sets $A_1, A_2, A_3$, respectively, according to agent $4$. Therefore, $v_4(A_1\setminus \{g_1\})>v_4(A_2\setminus \{g_2\})>v_4(A_3\setminus \{g_3\})$, and by comparing $v_4(A_4)$ with $v_4(A_1\setminus\{g_1\})$ we get that agent $4$ is  $\frac{1 + \varepsilon_1}{4 \beta -\varepsilon_4}$-\efo towards agent 1. The only thing that remains is to explore the possible deviations of agent 4. Initially, notice that regardless of what agent $4$ declares, she cannot get goods $g_1, g_2, g_3$ as these are taken in round 1 by the agents that precede her. With that in mind, we will examine what is the best attainable value through deviating, based on what she gets in round 1. Take note that she can get any goods from $\{g_4, g_5, \ldots, g_9\}$ in round 1 as they are available when her turn comes:
		
		\begin{itemize}
			\item\textbf{ Agent $4$ gets good $g_4$ in round 1.} Based on the reported preferences $\succ^*_1, \succ^*_2, \succ^*_3$ of the other agents, in round 2 we have the following: Good $g_5$ is lost to agent 1, good $g_7$ is lost to agent 2, and good $g_6$ to agent 3. Therefore, only goods $g_8$ and $g_9$ remain available for agent 4, and she can get only one of them. Thus, the maximum attainable value for her is $2 \beta +\varepsilon_1$.
			\item \textbf{Agent $4$ gets good $g_5$ in round 1.} In that case, based on the declaration of the rest of the agents, in round 2 we have the following: Good $g_4$ is lost to agent 1, good $g_7$ is lost to agent 2, and good $g_6$ to agent 3. Therefore, only goods $g_8$ and $g_9$ remain available for agent 4, and once more she can get only one of them. Thus, the maximum attainable value for her is $2 \beta-\varepsilon_4 +\varepsilon_1$.
			\item \textbf{Agent $4$ gets good $g_6$ in round 1.} Based on the reported preferences $\succ^*_1, \succ^*_2, \succ^*_3$ of the other agents, in round 2 we have the following: Good $g_5$ is lost to agent 1, good $g_7$ is lost to agent 2, and good $g_9$ to agent 3. Therefore, only goods $g_4$ and $g_9$ remain available for agent 4. Now observe that $v_4(g_4, g_6)= 2  \beta$ (as this is the value of the maximum matching), while $v_4(g_9, g_6)= 1+ \varepsilon_2$. Thus, the maximum attainable value for her is $2 \beta$.
			\item \textbf{Agent $4$ gets good $g_7$ in round 1.} Based on the reported preferences $\succ^*_1, \succ^*_2, \succ^*_3$ of the other agents, in round 2 we have the following: Good $g_5$ is lost to agent 1, good $g_4$ is lost to agent 2, and good $g_6$ to agent 3. Therefore, only goods $g_8$ and $g_9$ remain available for agent 4, and once more she can get only one of them. Thus, the maximum attainable value for her is $1-\varepsilon_3 +\varepsilon_1$.
			\item \textbf{Agent $4$ gets good $g_8$ in round 1.} Based on the reported preferences $\succ^*_1, \succ^*_2, \succ^*_3$ of the other agents, in round 2 we have the following: Good $g_5$ is lost to agent 1, good $g_7$ is lost to agent 2, and good $g_6$ to agent 3. Therefore, only goods $g_4$ and $g_9$ remain available for agent 4, and once more she can get only one of them. Thus, the maximum attainable value for her is $2  \beta + \varepsilon_1$.
			\item \textbf{Agent $4$ gets good $g_9$ in round 1.} In that case, based on the declaration of the rest of the agents, in round 2 we have the following: Good $g_5$ is lost to agent 1, good $g_7$ is lost to agent 2, and good $g_6$ to agent 3. Therefore, only goods $g_4$ and $g_8$ remain available for agent 4, and once more she can get only one of them. Thus, the maximum attainable value for her is $2  \beta +\varepsilon_2$.
			
		\end{itemize}
		
		From the above discussion we get that the maximum value that agent $4$ can attain through a deviation is $2 \cdot \beta + \varepsilon_1$. At the same time $v_4 (A_4)= 1+\varepsilon_1$. 
		By setting $\alpha = \frac{1 +\varepsilon_1}{2\cdot \beta +\varepsilon_1}$ we trivially have that $(\succ_1, \succ_2)$ is a $\alpha$-approximate PNE.
		On the other hand, for a given $\varepsilon > 0$, we have that $\frac{1 +\varepsilon_1}{2\cdot \beta +\varepsilon_1} + \varepsilon$ is strictly larger than $\frac{1 + \varepsilon_1}{4 \beta -\varepsilon_4}$ for sufficiently small $\varepsilon_1$. That is, there is a choice of $\varepsilon_1, \ldots, \varepsilon_6$ so that the $\alpha$-approximate PNE $(\succ^*_1, \succ^*_2, \succ^*_3, \succ_4)$ is not $\frac{\alpha}{2}+ \varepsilon$-\efo.
	\end{proof}

	\section{Discussion and Future Directions}
	\label{sec:discussion}
	
	In this work we studied the existence and fairness guarantees of the approximate pure Nash equilibria of the Round-Robin mechanism for agents with cancelable and submodular valuation functions. In both cases, 
	we generalized the surprising connection between the stable states of the mechanism and its fairness properties, a connection that was only known for exact equilibria and additive valuation functions. 
	For the function classes considered, we provide tight or almost tight bounds, thus giving a complete picture of the strengths and the limitations of the Round-Robin mechanism for these scenarios. 
	There are several interesting related directions, some of which we discuss below.
	
	An obvious first direction is
	to explore function classes beyond the ones studied here, with XOS or subadditive functions being prominent candidates. Since our results heavily rely on the properties of cancelable and submodular functions, it is likely that different approaches are needed for this endeavour.
	As we mention in the introduction, a second interesting direction, related to this one, is the study of the stability and fairness properties of variants of the Round-Robin mechanism that allow the agents to be more expressive. 
	Analyzing mechanisms that take as an input value oracles seems to be highly non-trivial, and although some of our results might transfer in this setting, we suspect that, in general, strong impossibility results hold regarding the fairness guarantees of approximate PNE.
	
	Finally, although here we focused on Round-Robin and \efo, most fair division algorithms  have not been considered in the strategic setting. One promising such algorithm, which is both fundamental in a number of  variants of the problem and simple enough, is the Envy-Cycle-Elimination algorithm of \citet{LMMS04} which is known to compute \efo allocations for general non-decreasing valuation functions. An appealing alternative here is studying the existence of equilibria of approximation algorithms for \mms allocations. An impoertant advantage in this case is that once the existence of an approximate PNE is shown, the corresponding \mms guarantee comes for free (see also the related discussion in Remark 2.9 of \citet{AmanatidisBFLLR21}).

	
	\bibliographystyle{my-plainnat}
	\bibliography{fairdivrefs}

\begin{thebibliography}{37}
\providecommand{\natexlab}[1]{#1}
\providecommand{\url}[1]{\texttt{#1}}
\expandafter\ifx\csname urlstyle\endcsname\relax
  \providecommand{\doi}[1]{doi: #1}\else
  \providecommand{\doi}{doi: \begingroup \urlstyle{rm}\Url}\fi

\bibitem[Akrami et~al.(2022)Akrami, Chaudhury, Garg, Mehlhorn, and
  Mehta]{akrami2022efx}
H.~Akrami, B.~R. Chaudhury, J.~Garg, K.~Mehlhorn, and R.~Mehta.
\newblock {EFX} allocations: Simplifications and improvements.
\newblock \emph{CoRR}, abs/2205.07638, 2022.
\newblock \doi{10.48550/arXiv.2205.07638}.
\newblock URL \url{https://doi.org/10.48550/arXiv.2205.07638}.

\bibitem[Amanatidis et~al.(2017{\natexlab{a}})Amanatidis, Birmpas,
  Christodoulou, and Markakis]{ABCM17}
G.~Amanatidis, G.~Birmpas, G.~Christodoulou, and E.~Markakis.
\newblock Truthful allocation mechanisms without payments: Characterization and
  implications on fairness.
\newblock In \emph{Proceedings of the 2017 {ACM} Conference on Economics and
  Computation, {EC' 17}}, pages 545--562. {ACM}, 2017{\natexlab{a}}.

\bibitem[Amanatidis et~al.(2017{\natexlab{b}})Amanatidis, Markakis, Nikzad, and
  Saberi]{AMNS17}
G.~Amanatidis, E.~Markakis, A.~Nikzad, and A.~Saberi.
\newblock Approximation algorithms for computing maximin share allocations.
\newblock \emph{{ACM} Trans. Algorithms}, 13\penalty0 (4):\penalty0
  52:1--52:28, 2017{\natexlab{b}}.

\bibitem[Amanatidis et~al.(2020)Amanatidis, Markakis, and Ntokos]{ANM2019}
G.~Amanatidis, E.~Markakis, and A.~Ntokos.
\newblock Multiple birds with one stone: Beating 1/2 for {EFX} and {GMMS} via
  envy cycle elimination.
\newblock \emph{Theor. Comput. Sci.}, 841:\penalty0 94--109, 2020.

\bibitem[Amanatidis et~al.(2021)Amanatidis, Birmpas, Fusco, Lazos, Leonardi,
  and Reiffenh{\"{a}}user]{AmanatidisBFLLR21}
G.~Amanatidis, G.~Birmpas, F.~Fusco, P.~Lazos, S.~Leonardi, and
  R.~Reiffenh{\"{a}}user.
\newblock Allocating indivisible goods to strategic agents: Pure {N}ash
  equilibria and fairness.
\newblock In \emph{Proceedings of the 17th International Conference on Web and
  Internet Economics, {WINE} 2021}, volume 13112 of \emph{LNCS}, pages
  149--166, 2021.

\bibitem[Amanatidis et~al.(2022)Amanatidis, Aziz, Birmpas, Filos{-}Ratsikas,
  Li, Moulin, Voudouris, and Wu]{AABFRLMVW22}
G.~Amanatidis, H.~Aziz, G.~Birmpas, A.~Filos{-}Ratsikas, B.~Li, H.~Moulin,
  A.~A. Voudouris, and X.~Wu.
\newblock Fair division of indivisible goods: {A} survey.
\newblock \emph{CoRR}, abs/2208.08782, 2022.
\newblock \doi{10.48550/arXiv.2208.08782}.
\newblock URL \url{https://doi.org/10.48550/arXiv.2208.08782}.

\bibitem[Aziz et~al.(2017{\natexlab{a}})Aziz, Bouveret, Lang, and
  Mackenzie]{AzizBLM17}
H.~Aziz, S.~Bouveret, J.~Lang, and S.~Mackenzie.
\newblock Complexity of manipulating sequential allocation.
\newblock In \emph{Proceedings of the 31st {AAAI} Conference on Artificial
  Intelligence, {AAAI '17}}, pages 328--334. {AAAI} Press, 2017{\natexlab{a}}.

\bibitem[Aziz et~al.(2017{\natexlab{b}})Aziz, Goldberg, and Walsh]{GW17}
H.~Aziz, P.~Goldberg, and T.~Walsh.
\newblock Equilibria in sequential allocation.
\newblock In \emph{Proceedings of the 5th International Conference on
  Algorithmic Decision Theory, {ADT '17}}, volume 10576 of \emph{LNCS}, pages
  270--283. Springer, 2017{\natexlab{b}}.

\bibitem[Aziz et~al.(2022)Aziz, Caragiannis, Igarashi, and Walsh]{AzizCIW22}
H.~Aziz, I.~Caragiannis, A.~Igarashi, and T.~Walsh.
\newblock Fair allocation of indivisible goods and chores.
\newblock \emph{Autonomous Agents and Multi Agent Systems}, 36\penalty0
  (1):\penalty0 3, 2022.

\bibitem[Babaioff et~al.(2021)Babaioff, Ezra, and Feige]{BabaioffEF21}
M.~Babaioff, T.~Ezra, and U.~Feige.
\newblock Fair and truthful mechanisms for dichotomous valuations.
\newblock In \emph{Proceedings of the 35th {AAAI} Conference on Artificial
  Intelligence, {AAAI} 2021}, pages 5119--5126. {AAAI} Press, 2021.

\bibitem[Barman and Krishnamurthy(2020)]{BarmanK20}
S.~Barman and S.~K. Krishnamurthy.
\newblock Approximation algorithms for maximin fair division.
\newblock \emph{{ACM} Trans. Economics and Comput.}, 8\penalty0 (1):\penalty0
  5:1--5:28, 2020.

\bibitem[Berger et~al.(2022)Berger, Cohen, Feldman, and Fiat]{BergerCFF22}
B.~Berger, A.~Cohen, M.~Feldman, and A.~Fiat.
\newblock Almost full {EFX} exists for four agents.
\newblock In \emph{Proceedings of the 36th {AAAI} Conference on Artificial
  Intelligence, {AAAI} 2022}, pages 4826--4833. {AAAI} Press, 2022.

\bibitem[Bouveret and Lang(2014)]{BouveretL14}
S.~Bouveret and J.~Lang.
\newblock Manipulating picking sequences.
\newblock In \emph{Proceedings of the 21st European Conference on Artificial
  Intelligence - {ECAI} 2014}, volume 263, pages 141--146. {IOS} Press, 2014.

\bibitem[Bouveret and Lema{\^{\i}}tre(2016)]{BL16}
S.~Bouveret and M.~Lema{\^{\i}}tre.
\newblock Characterizing conflicts in fair division of indivisible goods using
  a scale of criteria.
\newblock \emph{Autonomous Agents and Multi-Agent Systems}, 30\penalty0
  (2):\penalty0 259--290, 2016.

\bibitem[Bouveret et~al.(2016)Bouveret, Chevaleyre, and Maudet]{BCM16-survey}
S.~Bouveret, Y.~Chevaleyre, and N.~Maudet.
\newblock Fair allocation of indivisible goods.
\newblock In \emph{Handbook of Computational Social Choice}, pages 284--310.
  Cambridge University Press, 2016.

\bibitem[Budish(2011)]{Budish11}
E.~Budish.
\newblock The combinatorial assignment problem: Approximate competitive
  equilibrium from equal incomes.
\newblock \emph{Journal of Political Economy}, 119\penalty0 (6):\penalty0
  1061--1103, 2011.

\bibitem[Caragiannis et~al.(2009)Caragiannis, Kaklamanis, Kanellopoulos, and
  Kyropoulou]{CKKK09}
I.~Caragiannis, C.~Kaklamanis, P.~Kanellopoulos, and M.~Kyropoulou.
\newblock On low-envy truthful allocations.
\newblock In \emph{Proceedings of the 1st International Conference on
  Algorithmic Decision Theory, {ADT} 2009}, volume 5783 of \emph{LNCS}, pages
  111--119. Springer, 2009.

\bibitem[Caragiannis et~al.(2019)Caragiannis, Kurokawa, Moulin, Procaccia,
  Shah, and Wang]{CaragiannisKMPS19}
I.~Caragiannis, D.~Kurokawa, H.~Moulin, A.~D. Procaccia, N.~Shah, and J.~Wang.
\newblock The unreasonable fairness of maximum {N}ash welfare.
\newblock \emph{{ACM} Trans. Economics and Comput.}, 7\penalty0 (3):\penalty0
  12:1--12:32, 2019.

\bibitem[Caragiannis et~al.(2022)Caragiannis, Garg, Rathi, Sharma, and
  Varricchio]{caragiannis2022existence}
I.~Caragiannis, J.~Garg, N.~Rathi, E.~Sharma, and G.~Varricchio.
\newblock Existence and computation of epistemic efx allocations.
\newblock \emph{CoRR}, abs/2206.01710, 2022.
\newblock \doi{10.48550/arXiv.2206.01710}.
\newblock URL \url{https://doi.org/10.48550/arXiv.2206.01710}.

\bibitem[Chaudhury et~al.(2021)Chaudhury, Garg, and Mehta]{ChaudhuryGM21}
B.~R. Chaudhury, J.~Garg, and R.~Mehta.
\newblock Fair and efficient allocations under subadditive valuations.
\newblock In \emph{Proceedings of the 35th {AAAI} Conference on Artificial
  Intelligence, {AAAI} 2021}, pages 5269--5276. {AAAI} Press, 2021.

\bibitem[Foley(1967)]{Foley67}
D.~K. Foley.
\newblock Resource allocation and the public sector.
\newblock \emph{Yale Economics Essays}, 7:\penalty0 45--98, 1967.

\bibitem[Gamow and Stern(1958)]{GS58}
G.~Gamow and M.~Stern.
\newblock \emph{Puzzle-Math}.
\newblock Viking press, 1958.

\bibitem[Ghodsi et~al.(2022)Ghodsi, Hajiaghayi, Seddighin, Seddighin, and
  Yami]{GhodsiHSSY22}
M.~Ghodsi, M.~T. Hajiaghayi, M.~Seddighin, S.~Seddighin, and H.~Yami.
\newblock Fair allocation of indivisible goods: Beyond additive valuations.
\newblock \emph{Artif. Intell.}, 303:\penalty0 103633, 2022.

\bibitem[Halpern et~al.(2020)Halpern, Procaccia, Psomas, and Shah]{HPPS}
D.~Halpern, A.~D. Procaccia, A.~Psomas, and N.~Shah.
\newblock Fair division with binary valuations: One rule to rule them all.
\newblock In \emph{Proceedings of the 16th International Conference on Web and
  Internet Economics, {WINE} 2020}, volume 12495 of \emph{LNCS}, pages
  370--383. Springer, 2020.

\bibitem[Kurokawa(2017)]{Kurokawa17}
D.~Kurokawa.
\newblock \emph{Fair Division in Game Theoretic Settings}.
\newblock PhD thesis, Carnegie Mellon University, 2017.

\bibitem[Lehmann et~al.(2006)Lehmann, Lehmann, and Nisan]{LehmannLN06}
B.~Lehmann, D.~Lehmann, and N.~Nisan.
\newblock Combinatorial auctions with decreasing marginal utilities.
\newblock \emph{Games Econ. Behav.}, 55\penalty0 (2):\penalty0 270--296, 2006.

\bibitem[Leme(2017)]{Leme17}
R.~P. Leme.
\newblock Gross substitutability: An algorithmic survey.
\newblock \emph{Games Econ. Behav.}, 106:\penalty0 294--316, 2017.

\bibitem[Lipton et~al.(2004)Lipton, Markakis, Mossel, and Saberi]{LMMS04}
R.~J. Lipton, E.~Markakis, E.~Mossel, and A.~Saberi.
\newblock On approximately fair allocations of indivisible goods.
\newblock In \emph{Proceedings of the 5th {ACM} Conference on Electronic
  Commerce, EC '04}, pages 125--131. {ACM}, 2004.

\bibitem[Manurangsi and Suksompong(2021)]{ManurangsiS21}
P.~Manurangsi and W.~Suksompong.
\newblock Closing gaps in asymptotic fair division.
\newblock \emph{{SIAM} Journal on Discrete Mathematics}, 35\penalty0
  (2):\penalty0 668--706, 2021.

\bibitem[Markakis(2017)]{Markakis17-survey}
E.~Markakis.
\newblock Approximation algorithms and hardness results for fair division with
  indivisible goods.
\newblock In \emph{Trends in Computational Social Choice}, chapter~12. AI
  Access, 2017.

\bibitem[Markakis and Psomas(2011)]{MarkakisP11}
E.~Markakis and C.~Psomas.
\newblock On worst-case allocations in the presence of indivisible goods.
\newblock In \emph{Proceedings of the 7th International Conference on Web and
  Internet Economics, {WINE} 2011}, volume 7090 of \emph{LNCS}, pages 278--289.
  Springer, 2011.

\bibitem[Nemhauser et~al.(1978)Nemhauser, Wolsey, and Fisher]{NemhauserWF78}
G.~L. Nemhauser, L.~A. Wolsey, and M.~L. Fisher.
\newblock An analysis of approximations for maximizing submodular set functions
  - {I}.
\newblock \emph{Math. Program.}, 14\penalty0 (1):\penalty0 265--294, 1978.

\bibitem[Plaut and Roughgarden(2020)]{PR18}
B.~Plaut and T.~Roughgarden.
\newblock Almost envy-freeness with general valuations.
\newblock \emph{{SIAM} J. Discret. Math.}, 34\penalty0 (2):\penalty0
  1039--1068, 2020.

\bibitem[Procaccia(2016)]{Procaccia16-survey}
A.~D. Procaccia.
\newblock Cake cutting algorithms.
\newblock In \emph{Handbook of Computational Social Choice}, pages 311--330.
  Cambridge University Press, 2016.

\bibitem[Psomas and Verma(2022)]{PV22}
A.~Psomas and P.~Verma.
\newblock Fair and efficient allocations without obvious manipulations.
\newblock In \emph{Advances in Neural Information Processing Systems 35: Annual
  Conference on Neural Information Processing Systems 2022, {NeurIPS} 2022},
  2022.

\bibitem[Steinhaus(1949)]{Steinhaus49}
H.~Steinhaus.
\newblock Sur la division pragmatique.
\newblock \emph{Econometrica}, 17 (Supplement):\penalty0 315--319, 1949.

\bibitem[Varian(1974)]{Varian74}
H.~R. Varian.
\newblock Equity, envy and efficiency.
\newblock \emph{Journal of Economic Theory}, 9:\penalty0 63--91, 1974.

\end{thebibliography}

\end{document}